\documentclass[a4paper,UKenglish,cleveref, autoref]{lipics-v2019}
\usepackage{subcaption}
\usepackage{tikz}
\usepackage{multirow}
\usepackage{algorithm}
\usepackage[noend]{algpseudocode}

\title{On Existence of Equilibrium Under Social Coalition Structures} %TODO Please add

\titlerunning{On Existence of Equilibrium Under Social Coalition Structures}%optional, please use if title is longer than one line

\author{Bugra Caskurlu}{TOBB University of Economics and Technology, Ankara, Turkey.}{bcaskurlu@etu.edu.tr}{}{}

\author{Ozgun Ekici}{Ozyegin University, Istanbul, Turkey.}{ozgun.ekici@ozyegin.edu.tr}{}{}

\author{Fatih Erdem Kizilkaya}{TOBB University of Economics and Technology, Ankara, Turkey.}{f.kizilkaya@etu.edu.tr}{}{}

\funding{This work is supported by The Scientific and Technological Research Council of Turkey (T\"UB\.ITAK) through grant 118E126.}

\authorrunning{Caskurlu et al.}%TODO mandatory. First: Use abbreviated first/middle names. Second (only in severe cases): Use first author plus 'et al.'

%TODO mandatory: Please choose ACM 2012 classifications from https://dl.acm.org/ccs/ccs_flat.cfm

%TODO mandatory; please add comma-separated list of keywords

\nolinenumbers %uncomment to disable line numbering

\hideLIPIcs  %uncomment to remove references to LIPIcs series (logo, DOI, ...), e.g. when preparing a pre-final version to be uploaded to arXiv or another public repository

\acknowledgements{We would like to thank Elliot Anshelevich for his valuable suggestions and insights.}

\begin{document}

\maketitle

%TODO mandatory: add short abstract of the document
\begin{abstract}
In a strategic form game a strategy profile is an equilibrium if no viable
coalition of agents (or players) benefits (in the Pareto sense) from jointly
changing their strategies. Weaker or stronger equilibrium notions can be
defined by considering various restrictions on coalition formation. In a Nash
equilibrium, for instance, the assumption is that viable coalitions are
singletons, and in a super strong equilibrium, every coalition is viable.
Restrictions on coalition formation can be justified by communication
limitations, coordination problems or institutional constraints. In this
paper, inspired by social structures in various real-life scenarios, we
introduce certain restrictions on coalition formation, and on their basis we
introduce a number of equilibrium notions. As an application we study our
equilibrium notions in resource selection games (RSGs), and we present a
complete set of existence and non-existence results for general RSGs and their
important special cases.
\end{abstract}

\newpage

\section{Introduction}
\label{sec:Introduction}

In game theory the centerpiece of analysis is the
notion of an equilibrium. In a game in strategic form, an equilibrium is a
strategy profile at which certain types of coalitions of agents
do not have profitable deviations. The strongest notion that can be defined
along this line is a super strong equilibrium: no coalition of agents benefits (in the Pareto sense) from jointly changing their strategies. Note that in a game with $n$ agents there are as many as $2^{n}-1$ (non-empty) possible coalitions if any coalition is deemed viable. However, deeming every coalition viable and disqualifying strategy profiles as non-equilibrium may be misguided. First of all, a super strong equilibrium rarely exists in a game. Therefore, restrictions on coalition formation may be helpful to obtain existence results\footnote{In defining our equilibrium notions we use the weak domination relation: a deviation makes coalition members better off in the Pareto sense. An alternative approach is to define an equilibrium using the strong domination relation: a deviation makes \emph{every} coalition member strictly better off. Even when the strong domination relation is used, an equilibrium rarely exists in a game if every coalition is viable (the so-called strong equilibrium notion). For studies on strong equilibrium, its existence, and some other related work, see \cite{Ref_Aumann1959, Ref_BPW1997, Ref_HL1997, Ref_KLW1997, Ref_KLW1999, Ref_MW1996}.}.

This is the very same idea behind the well-known Nash equilibrium \cite{Ref_N1951} solution concept where only singletons are viable coalitions. In this paper, our goal is to fill the gap between the less restrictive Nash equilibrium notion and the very restrictive super strong equilibrium notion. The restrictions we enforce upon coalitions are not merely mathematical generalizations, but also motivated by many real-life examples.

Coalition formation may be restricted by coordinational, communicational and institutional constraints.

\begin{itemize}

\item[-] \textbf{Coordinational:} A deviation by a coalition requires coalition members to act in unison. However, if coalition members are not familiar with one another, taking coordinated action becomes difficult. Or, everyone may be familiar with one another yet agents may find it more difficult to coordinate as the number of coalition members grows.

\item[-] \textbf{Communicational:} Formation of coalitions may require private communication. For instance, imagine that agents communicate through a network where each agent is located at one of the nodes. If some agent $i$ wants to offer a deviation to another agent $j$, then agent $i$ had better make sure that his offer does not deteriorate any of the agents along the path, since otherwise it will probably not reach agent $j$.

\item[-] \textbf{Institutional:} Even if there does not exist any coordinational or communicational barrier between two agents to form a coalition, there might exists self-imposed institutional constraints. In global affairs, it is not uncommon that a government feels compelled to act in unison with its allies even if doing so comes at a great cost. For instance, it may be forced to uphold trade sanctions on a neighboring country, causing much harm on its economy. Or, a nation may refuse to engage in mutually beneficial relations with another nation due to historical enmities.
\end{itemize}

Note that a full consideration of what restrictions on coalition formation may be reasonable in a specific real-life scenario is beyond the scope of our paper. We rather focus on restrictions that are motivated by natural real-life social structures that may arise in various settings.

On the basis of our restrictions, we define new equilibrium notions and then study how they relate to one another, and when they are guaranteed to exist. Adding social structures to games is actually a growing trend in the recent literature. The following equilibrium notion introduced in an earlier study is related to our study in particular\footnote{For two other related studies see Ashlagi et al. \cite{Ref_AKT2008} and Hoefer et al. \cite{Ref_HPPSV2011}.}:

\begin{itemize}
\item[-] \textbf{Partition Equilibrium:} In a partition equilibrium, it is assumed that the set of viable coalitions is a partition of the set of agents; see Figure \ref{fig:partition}. This notion generalizes the notion of a Nash equilibrium and has been introduced by Feldman and Tennenholtz \cite{Ref_FT2009}.
\end{itemize}

\begin{figure}[h]
\centering
\begin{subfigure}[b]{.5\textwidth}
\centering
\scalebox{.85}{
\begin{tikzpicture}
\draw (-25pt, 0pt) ellipse (20pt and 12pt);
	\draw (-25pt, 0pt) node{1 \quad 2};
\draw (25pt, 0pt) ellipse (20pt and 12pt);
	\draw (25pt, 0pt) node{3 \quad 4};
\draw (80pt, 0pt) ellipse (28pt and 15pt);
	\draw (80pt, 0pt) node{5 \quad 6 \quad 7};
\draw (135pt, 0pt) ellipse (20pt and 12pt);	
	\draw (135pt, 0pt) node{8 \quad 9};	
\end{tikzpicture}}
\caption{\centering A Partition Coalition Structure}
\label{fig:partition}
\end{subfigure}\begin{subfigure}[b]{.5\textwidth}
\centering
\scalebox{.85}{
\begin{tikzpicture}
\draw (0pt, 0 pt) ellipse (50pt and 20pt);
	\draw (-25pt, 0pt) ellipse (20pt and 10pt);
		\draw (-25pt, 0pt) node{1 \quad 2};
	\draw (25pt, 0pt) ellipse (20pt and 10pt);
		\draw (25pt, 0pt) node{3 \quad 4};
		
\draw (120pt, 0pt) ellipse (50pt and 20pt);
\draw (120pt, 0pt) node{5 \quad 6 \quad 7 \quad 8 \quad 9};
	\draw (138pt, 0pt) ellipse (28pt and 11pt);
\end{tikzpicture}}
\caption{\centering A Laminar Coalition Structure}
\label{fig:laminar}
\end{subfigure}
\caption{Examples of Partition and Laminar Coalition Structures}%
\end{figure}
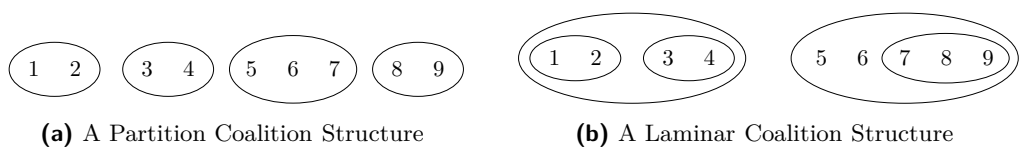

\medskip

Along similar lines we introduce in our paper three new notions of equilibrium, motivated by real-life social structures:

\medskip

\begin{itemize}
\item[-] \textbf{Laminar Equilibrium:} In a laminar equilibrium, it is assumed that the set of viable coalitions exhibits a laminar structure; see Figure \ref{fig:laminar}. This notion is mainly motivated by institutional constraints as it relates to hierarchical communities in real life. For instance, a military is divided into corps, legions, and brigades; a cabinet is divided into ministries, departments, and directorates; a university is divided into faculties and departments; and a company is divided into business units, divisions, and departments.

\item[-] \textbf{Contiguous Equilibrium:} In a contiguous equilibrium, it is assumed that agents are distributed on a line and each viable coalition consists of some agents that are ordered on the line subsequently; see Figure \ref{fig:contiguous}. A contiguous coalition structure may emerge in real life due to coordinational, communicational and institutional constraints as depicted in following scenarios, respectively: $i)$ Residents of a street are most likely to socialize via neighbourhood, $ii)$ When private communication between players are restricted in an environment such as a queue, $iii)$ When agents are positioned on left-right political spectrum, coalitions presumably cannot be formed without intermediaries.

\item[-] \textbf{Centralized Equilibrium:} In a centralized equilibrium, it is assumed that agents are distributed on a plane and each viable coalition corresponds to a circle on the plane such that a coalition member lies at the circle's center and the agents that lie inside the circle are the coalition members; see Figure \ref{fig:centralized}. A centralized coalition structure may emerge in real life due to coordinational, communicational and institutional constraints as depicted in following scenarios, respectively: $i)$ Residents of a neighbourhood are more likely to socialize inside a closed distance, $ii)$ When agents can only communicate within a specific distance due to various reasons, such as wireless coverage, $iii)$ When agents are positioned on a political compass, the radius of a coalition corresponds to the tolerance of its center (possibly the leader) to other political views.
\end{itemize}

\begin{figure}[h]
\centering
\begin{subfigure}[b]{.5\textwidth}
\centering
\scalebox{.85}{
\begin{tikzpicture}
\draw (26pt, 0pt) node{1 \qquad 2 \qquad 3 \qquad 4 \qquad 5 \qquad 6 };	
\draw (0pt, 0 pt) ellipse (58pt and 20pt);
\draw (-14pt, 0pt) ellipse (37pt and 12pt);
\draw (83pt, 0pt) ellipse (24pt and 10pt);
\draw (50pt, 0 pt) ellipse (58pt and 20pt);
\end{tikzpicture}}
\caption{\centering A Contiguous Coalition Structure}
\label{fig:contiguous}
\end{subfigure}\begin{subfigure}[b]{.5\textwidth}
\centering
\scalebox{.85}{
\begin{tikzpicture}
\fill (0pt, 30pt) circle (2pt) node[anchor = south east]{1};
\fill (30pt, 30pt) circle (2pt) node[anchor = south east]{2};
\fill (0pt, 0pt) circle (2pt) node[anchor = north east]{3};
\fill (30pt, 0pt) circle (2pt) node[anchor = north east]{4};
\fill (56pt, 15pt) circle (2pt) node[anchor = east]{5};
\fill (86pt, 15pt) circle (2pt) node[anchor = west]{6};
\fill (112pt, 30pt) circle (2pt) node[anchor = south west]{7};
\fill (142pt, 30pt) circle (2pt) node[anchor = south west]{8};
\fill (112pt, 0pt) circle (2pt) node[anchor = north west]{9};
\fill (142pt, 0pt) circle (2pt) node[anchor = north west]{10};
\draw (0pt, 0pt) circle (30pt);
\draw (30pt, 0pt) circle (30pt);
\draw (0pt, 30pt) circle (30pt);
\draw (30pt, 30pt) circle (30pt);
\draw (56pt, 15pt) circle (30pt);
\draw (86pt, 15pt) circle (30pt);
\draw (112pt, 30pt) circle (30pt);
\draw (142pt, 30pt) circle (30pt);
\draw (112pt, 0pt) circle (30pt);
\draw (142pt, 0pt) circle (30pt);
\end{tikzpicture}}
\caption{\centering A Centralized Coalition Structure}
\label{fig:centralized}
\end{subfigure}
\caption{Examples of Contiguous and Centralized Coalition Structures}%
\end{figure}
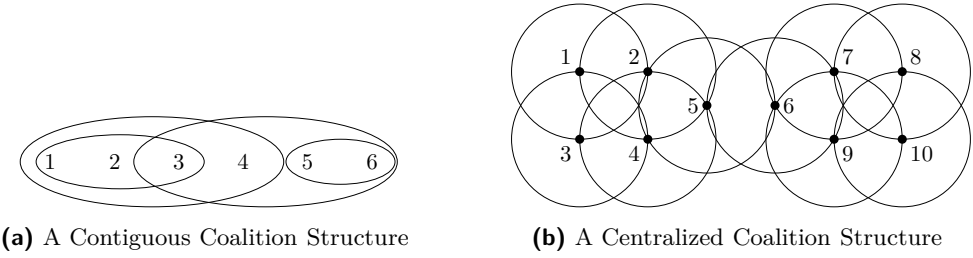

Notice that the number of viable coalitions is $O(n)$ in the case of a partition equilibrium or a laminar equilibrium, and it is $O(n^{2})$ in the case of a contiguous equilibrium or a centralized equilibrium (where $n$ is the number of agents). However, the number of possible coalition structures w.r.t. the number of agents for these notions are beyond the scope of our paper\footnote{Partition problem was famously solved by Ramanujan; however, the number of possible laminar families has been an open problem for years.}. In Theorem 1, we show that each equilibrium notion above generalizes the preceding one.

\bigskip

As an application we study the existence of the above notions of equilibrium in resource selection games (RSGs), for the following reasons:

\begin{itemize}
\item[-] RSGs fall into the class of potential games for which the existence of a Nash equilibrium is guaranteed (see \cite{Ref_MS1996, Ref_R1973}). Since the newly defined solution concepts are generalizations of Nash equilibrium, existence of equilibria w.r.t. them is not guaranteed in classes of games for which Nash equilibrium is not guaranteed to exist.

\item[-] However, super strong equilibrium does not exist even in the simplest special cases of this class of games. Hence, it is not trivial whether existence of equilibria w.r.t. the above solution concepts is guaranteed or not.

\item[-] RSGs are a subclass of congestion games \cite{Ref_R1973} which has immense number of applications \cite{Ref_HTW2006, Ref_HL1997, Ref_KS2014, Ref_M1996}. Simple as they may be, RSGs capture the essence of various games especially in the domain of routing games. In this setting, they are mostly known as parallel-link networks. For recent literature on parallel-link networks, see \cite{esa} and the references therein.

\item[-] Aside from their natural applications in transportation and communication networks RSGs have been also shown to be useful in biology \cite{Ref_M1979, Ref_QM1994}.

\item[-] Not only the immediate previous work \cite{Ref_ACH2013, Ref_FT2009}, but also several other newly defined solution concepts \cite{Ref_HTW2006, Ref_HPPSV2011} were studied for RSGs. So, RSGs are a benchmark to study existence of equilibrium w.r.t. newly defined solution concepts (\cite{Ref_HTW2006} uses the term symmetric load balancing games for RSGs).
\end{itemize}

Our results in RSGs and their relation to the results in the literature are as follows: Feldman and Tennenholtz \cite{Ref_FT2009} showed that a partition equilibrium always exists in RSGs under the following restrictions: (i) if the size of a viable coalition is bounded by $2$; or (ii) if there are only two resources; or (iii) if the resources are identical. Anshelevich et al. \cite{Ref_ACH2013} generalized this result by proving a strategy profile that is both a partition equilibrium and a Nash equilibrium is guaranteed to exist in general RSGs. \textbf{Our findings are as follows:}

\begin{itemize}
\item[-] In Section 3.1 and 3.2, we generalize the results (ii) and (iii) above in \cite{Ref_FT2009} to the notion of a laminar equilibrium. We prove that a laminar equilibrium always exists: $i)$ If there are only two resources (Theorem \ref{teo_two_yesLE}), or $ii)$ If the resources are identical (Corollary \ref{cor:iden_yesLE}). Note that RSGs with two resources is interesting in its own right. For instance, the well-known $PoA = 4/3$ result for selfish routing also holds for parallel-link networks with two links \cite{correa}.

\item[-] In Section 3.1, we show that an analogous generalization of the result in \cite{Ref_ACH2013} is not possible. Via an intricate counterexample, we show that a laminar equilibrium may not exist in general RSGs (Theorem \ref{teo_noLE}). Indeed, our counterexample shows that in general RSGs there may not exist a strategy profile that is Pareto efficient, a partition equilibrium, and a Nash equilibrium (Corollary \ref{cor:NashParetoPartition}). Notice that the main existence result in \cite{Ref_ACH2013} does not survive a minimal extension of their domain of viable set of coalitions, i.e., when the set of all agents is added to the viable set of coalitions.

\item[-] In Section 3.2, we prove that a contiguous equilibrium may not always exist in an RSG with two resources (Theorem \ref{teo_two_noCOE}). We show that, however, a contiguous equilibrium always exists when resources are identical (Theorem \ref{teo_ident_yesCoE}).

\item[-] In Section 3.2, we show that, however, a centralized equilibrium may not exist even in the very special setting in which there are two identical resources (Theorem \ref{teo_twoident_noCEE}). In the two identical resources setting, Feldman and Tennenholtz \cite{Ref_FT2009} showed that a super strong equilibrium may not exist.
\end{itemize}

\textbf{Table \ref{Table_ExResults} below summarizes these findings:}

\begin{table}[h]
\centering
\begin{tabular}
[b]{||c||cccc||}\hline
\multirow{2}{*}{$Solution \ Concepts$} & \multicolumn{4}{c||}{$Resources$%
}\\ \cline{2-5}
& {\small $general$} & {\small $two$} & {\small $identical$} & {\small $two
\ identical$}\\ \hline
$Partition$ & $+ ^{**}$ & $+ ^{*}$ & $+ ^{*}$ & $+ ^{*}$\\ \hline
$Laminar$ & $-$ & $+$ & $+ $ & $+$\\ \hline
$Contiguous$ & $-$ & $-$ & $+$ & $+$\\ \hline
$Centralized$ & $-$ & $-$ & $-$ & $-$\\ \hline
\multicolumn{5}{||c||}{{\footnotesize $^{*}$ due to Feldman and Tennenholtz
\cite{Ref_FT2009} \  \  \  \ $^{**}$ due to Anshelevich et al. \cite{Ref_ACH2013}%
}}\\ \hline
\end{tabular}
\caption{Existence and Non-existence Results}%
\label{Table_ExResults}%
\end{table}

\section{The Equilibrium Notions}
\label{sec:equinot}

This section introduces our equilibrium notions in the context of a strategic form game and then studies how these notions are related.

\medskip

Let $\langle N,S,U\rangle$ be a \emph{strategic form game} where $N$ is a finite set of \emph{agents} (or players), $S:(S_{j})_{j\in N}$ is the \emph{strategy space} and $U: S \rightarrow \mathbb{R}^{|N|}$ is the \emph{payoff function.} Agent $j$'s payoff at strategy profile $s\in S$ is denoted by $U_{j}%
(s)$.\footnote{Throughout, $\subset$ and $\subseteq$ denote the
\textquotedblleft strict subset of\textquotedblright \ and the
\textquotedblleft subset of\textquotedblright \ relations. For a set $X$, $|X|$ denotes the cardinality of $X$. For a number $x$, $|x|$ denotes the absolute value of $x$,\ and $\lfloor x\rfloor$ denotes the greatest integer smaller than $x$.} A \emph{coalition} $c$ is a non-empty subset of agents. Let $\mathcal{P}(N)$
be the power set of $N$. Then the domain of coalitions is $\mathcal{P}%
(N)-\{ \emptyset \}$. Let $\mathcal{P}_{\geq1}(N)$ denote this domain. A \emph{coalition structure} $C$ is a set of viable coalitions; i.e.,
$C\subseteq \mathcal{P}_{ \geq1}(N)$.

\medskip

Let $S_{c}$ denote the restriction of the strategy space for coalition $c$.
Let $s_{c}$ denote the restriction of the strategy profile $s$ for coalition
$c$. That is, $S_{c}=(S_{j})_{j\in c}$ and $s_{c}=(s_{j})_{j\in c}$. Note that the strategy space can be written as $(S_{c},S_{N\smallsetminus c}%
)$. The space $S_{c}$ represents the domain of \emph{deviations} for coalition
$c$. At $s$ if coalition $c$ takes deviation $\widetilde{s}_{c}\in S_{c}$, the
resulting strategy profile is $(\widetilde{s}_{c},s_{N\smallsetminus c}%
)\in(S_{c},S_{N\smallsetminus c})$. This is a \emph{profitable deviation} for
coalition $c$ if for each $j\in c$, $U_{j}(\widetilde{s}_{c}%
,s_{N\smallsetminus c})\geq U_{j}(s)$, and for some $j\in c$, $U_{j}%
(\widetilde{s}_{c},s_{N\smallsetminus c})>U_{j}(s)$. That is, the deviation
makes coalition $c$ better off in the Pareto sense. A strategy profile $s$ is called $c$\emph{-stable} if coalition $c$ has no profitable deviation at $s$, and $C$\emph{-stable} if for coalition structure $C$, $s$ is
$c$\emph{-stable} for each $c\in C$.

\medskip

Notice that a strategy profile is a
\emph{super strong equilibrium} if it is $\mathcal{P}_{\geq1}(N)$-stable, and a strategy profile is a \emph{Nash
equilibrium} if it is $\mathcal{P}_{=1}(N)$-stable where $\mathcal{P}%
_{=1}(N)=\{c\subset N|$ $\left \vert c\right \vert =1\}$. We now define the partition equilibrium which was introduced in the earlier literature, and the three notions of equilibrium which are introduced first in our paper.

\begin{itemize}
\item[-] \textbf{Partition Equilibrium:} A coalition structure $C$ is a
\emph{partition} if for each $j\in N$, there exists a unique coalition $c\in
C$ such that $j\in c$. Given a partition coalition structure $C$, a strategy
profile is a \emph{partition equilibrium} if it is $C$-stable.

\item[-] \textbf{Laminar Equilibrium}: A coalition structure $C$ is
\emph{laminar} if for any two coalitions $c_{1},c_{2}\in C$ such that
$c_{1}\cap c_{2}\neq \emptyset$, either $c_{1}\subseteq c_{2}$ or
$c_{2}\subseteq c_{1}$. Given a laminar coalition structure $C$, a strategy
profile is a \emph{laminar equilibrium} if it is $C$-stable.

\item[-] \textbf{Contiguous Equilibrium:} A coalition structure $C$ is
\emph{contiguous} if there exists a path $P:j_{1}-j_{2}-\cdots-j_{\left \vert
N\right \vert }$ (the vertices are agents) in accordance with $C$ in the
following sense: for each $c\in C$, the agents in $c$ are subsequently ordered
under $P$. Given a contiguous coalition structure $C$, a strategy profile is a
\emph{contiguous equilibrium} if it is $C$-stable.

\item[-] \textbf{Centralized Equilibrium:} A coalition structure $C$ is
\emph{centralized} if there exists a planar representation $(\phi,\psi)$,
where $\phi:N\rightarrow \mathbb{R}^{2}$ and $\psi:C\rightarrow(N\times
\mathbb{R}_{>0})$, which is in accordance with $C$ in the following sense:

\begin{itemize}
\item[-] For each $c\in C$, in the Cartesian space, $\psi \left(  c\right)  $
corresponds to the following circle: the circle's center is at point
$\phi(\psi_{1}(c))$ and its radius is $\psi_{2}(c)$.

\item[-] For an agent $j$, $\phi(j)$ lies inside the circle corresponding to
$\psi \left(  c\right)  $ (the boundary included) if and only if $j\in c$.
\end{itemize}

In simpler terms, agents lie on a plane and a viable coalition consists of
agents that lie inside a circle with the restriction that one coalition member
lies at the circle's center. Given a centralized coalition structure $C$, a
strategy profile is a \emph{centralized equilibrium} if it is $C$-stable.
\end{itemize}

Let $\mathcal{C}^{sse}=\left \{  \mathcal{P}_{\geq1}(N)\right \}  $. Let
$\mathcal{C}^{ne}=\{ \mathcal{P}_{=1}(N)\}$. Also, let $\mathcal{C}%
^{pe},\mathcal{C}^{le},\mathcal{C}^{coe},\mathcal{C}^{cee}$ be, respectively,
the domains of coalition structures that are partitions, laminar, contiguous,
and centralized. Thus, a strategy profile that is $C$-stable is a super strong
equilibrium if $C\in \mathcal{C}^{sse}$; a Nash equilibrium if $C\in
\mathcal{C}^{ne}$; a partition equilibrium if $C\in \mathcal{C}^{pe}$, and so on.

\medskip

Recall that our equilibrium notions are inspired by various real-life social structures. So it could well be the case that they are not interrelated. However, we show that each equilibrium notion above generalizes the preceding one, in Theorem \ref{teo_mostrefined_SSE}, the proof of which appears in Appendix A.

\begin{theorem}
\label{teo_mostrefined_SSE} We have $\mathcal{C}^{ne}\subseteq \mathcal{C}^{pe}\subseteq \mathcal{C}^{le}\subseteq \mathcal{C}^{coe}\subseteq \mathcal{C}^{cee}$. Also,

\begin{itemize}
\item[-] $\mathcal{C}^{ne}\subset \mathcal{C}^{pe}\subset \mathcal{C}^{le}$ for $|N|\geq 2$,
\item[-] $\mathcal{C}^{le}\subset \mathcal{C}^{coe}$ for $|N|\geq 3$,
\item[-] $\mathcal{C}^{coe}\subset \mathcal{C}^{cee}$ for $|N|\geq 4$,
\item[-] for each $C\in \mathcal{C}^{ne}\cup \mathcal{C}^{pe}\cup \mathcal{C}^{le}\cup \mathcal{C}^{coe}\cup \mathcal{C}^{cee}$, $C\subseteq \mathcal{P}_{\geq1}(N)$,
\item[-]$\mathcal{P}_{\geq1}(N)\notin \mathcal{C}^{ne}\cup \mathcal{C}^{pe}\cup \mathcal{C}^{le}\cup \mathcal{C}^{coe}\cup \mathcal{C}^{cee}$ for $|N|\geq 3$.
\end{itemize}

\medskip

That is, centralized equilibrium is a generalization of contiguous equilibrium, contiguous equilibrium is a generalization of laminar equilibrium, laminar equilibrium is a generalization of partition equilibrium, partition equilibrium is a generalization of Nash equilibrium. Super strong equilibrium is a refinement of all these equilibrium notions. The generalizations and the refinement are nontrivial for $|N|\geq 4$.
\end{theorem}

\section{An Application: Resource Selection Games}
\label{sec:applications}

A \emph{resource selection game} (RSG) is a triplet $\langle
N,M,f\rangle$ where $N:\{1, 2, \ldots, n\}$ is the set of \emph{agents}, $M:\{1, 2, \ldots, m\}$ is the set of \emph{resources} and $f:(f_{i})_{i = 1}^{m}$ is the profile of strictly monotonic increasing \emph{cost functions} such that $f_{i}(0) = 0$ for all $i \in \{1,2, \ldots, m\}$. When $q$ agents use resource $i$, each incurs a cost equal to $f_{i}(q)$. Each agent tries to minimize the cost it incurs. In the rest of the paper we fix the game $\langle N,M,f\rangle$.

\medskip

An \emph{allocation} is a sequence $a:(a_{i})_{i=1}^{m}$ such that: $i)$ For each $i\in M$, we have $a_{i}\subseteq N$; $ii)$ For every $i, i ^{\prime}\in M$ such that $i\neq i^{\prime}$, we have $a_{i} \cap a_{i^{\prime}} = \emptyset$; and $iii)$ We have $\bigcup_{i\in M} a_{i} = N$. Above, $a_{i}$ denotes the set of agents that are assigned to resource $i$ at allocation $a$. Thus, at allocation $a$, each agent in $a_{i}$ incurs a cost equal to $f_{i}(|a_{i}|)$. Let $\mathcal{A}$ be the domain of allocations.

The \emph{maxcost} of an allocation $a$ is the maximum cost incurred by an
agent at $a$. That is, the maxcost of allocation $a$ equals $max_{i\in M}$
$f_{i}(|a_{i}|)$. The \emph{minmaxcost} of the RSG, to be denoted by $\alpha$, is the maxcost of
the allocation whose maxcost is smallest. That is, $\alpha=min_{a\in
\mathcal{A}}$ $max_{i\in M}$ $f_{i}(|a_{i}|)$.

Let $q_{i}=max_{q\in \mathbb{Z}_{\geq0}}$ $f_{i}(q)\leq \alpha$. We refer to
$q_{i}$ as resource $i$'s quota. That is, a resource's quota is the maximum
number of agents that can be assigned to it without making its cost exceed
$\alpha$. We distinguish between resources which can and cannot attain the
minmaxcost $\alpha$. A resource $i$ is a \emph{Type 1 resource} if $f_{i}(q_{i}) = \alpha$, and a \emph{Type 2 resource} if $f_{i}(q_{i}) < \alpha$.

Let $T_{1}$ and $T_{2}$ denote, respectively, the sets of type 1 and type 2
resources. Since the minmaxcost of the game is $\alpha$, we have $T_{1}%
\neq \emptyset$. Also, for $i\in T_{1}$, let $\beta_{i}=f_{i}(q_{i}-1)$. We
refer to $\beta_{i}$ as resource $i$'s \emph{beta value}. Note that for a type
1 resource $i$, its beta value is its cost when the number of agents assigned
to it is one less than its quota.

\medskip

Note that an RSG is a non-cooperative game in the strategic form although its
formulation here is different from the formulation of a strategic form game in
Section 2. Here, agents' payoffs are negative (i.e., they incur costs rather
than receive payoffs) and an agent's strategy space is the set $M$ (i.e., the
agent selects one of the resources).

In this context, we continue to use the terminology in Section 2 in regards to
coalitions and coalition structures; i.e., $c,C$, $\mathcal{P}_{=1}(N)$,
$\mathcal{P}_{\geq1}(N)$, $\mathcal{C}^{sse}$, $\mathcal{C}^{ne}$,
$\mathcal{C}^{le}$, $\mathcal{C}^{coe}$, $\mathcal{C}^{cee}$ are as described
in Section 2. We also use the terminology in Section 2 regarding the stability
and equilibrium notions but with one exception: Note that in an RSG an
allocation fully specifies the strategies of agents. Therefore, in this
context we speak of an \textquotedblleft allocation\textquotedblright \ as a
substitute for a strategy profile. Hence, in this context, rather than a
strategy profile we speak of an allocation being $c$-stable or $C$-stable; or
being a laminar equilibrium or a contiguous equilibrium.

Also, in this context, we represent a deviation by a coalition $c$ as a
sequence $(c_{i})_{i=1}^{m}$ such that: (i) $c_{1}\cup c_{2}\cup \cdots \cup
c_{m}=c$; and (ii) for each $i,i^{\prime}\in M$ and $i\neq i^{\prime}$,
$c_{i}\cap c_{i^{\prime}}=\emptyset$. That is, a deviation is an agreement by
coalition members on which resources they will use: $c_{i}$ is the set of
coalition members who agree to use resource $i$. We use $a\circ(c_{i}%
)_{i=1}^{m}$ to denote the allocation that results when coalition $c$ takes
deviation $(c_{i})_{i=1}^{m}$ at allocation $a$: i.e., after the deviation the
set of agents that are assigned to resource $i$ is $(a_{i}\smallsetminus
c)\cup c_{i}$. Also, note that a deviation is a profitable deviation if at the
resulting allocation each coalition member becomes \emph{weakly better off}
(i.e., the cost it incurs does not increase) and at least one of them becomes
\emph{better off} (i.e., the cost it incurs decreases).

\medskip

The notion of a super strong equilibrium is very appealing since it precludes
profitable deviations by any coalition of agents. However, in most game forms
a super strong equilibrium is not guaranteed to exist. The is also true for
RSGs; see the Example \ref{ex_noSSE} below.

\begin{example}[due to Feldman and Tennenholtz \cite{Ref_FT2009}]
\label{ex_noSSE} Consider
the RSG where $N=\{1,2,3\}$, $M=\{1,2\}$, and $f_{i}(q_{i})=q_{i}$ for $i=1,2$.

\medskip

In this RSG there exists no super strong equilibrium. To see this note that:
At an allocation where all agents are assigned to the same resource, an agent
that deviates to the other resource becomes better off. In all other
allocations, two agents are assigned to one of the resources and one agent is
assigned to the other resource. Wlog., let agents $1$ and $2$ be assigned to
resource $1$ and agent $3$ to resource $2$. But now the coalition $c=\{1,2\}$
has a profitable deviation: When agent $1$ deviates to resource $2$, agent $1$
becomes weakly better off and agent $2$ becomes better off. \hfill$\Diamond$
\end{example}

We next present a characterization of Nash equilibrium in RSGs given by \cite{Ref_ACH2013}.

\begin{theorem}[due to Anshelevich et al. \cite{Ref_ACH2013}]
\label{thm:ExistenceNash}
In RSGs there always exists a
Nash equilibrium allocation. An allocation $a$ is a Nash equilibrium if and
only if:
\begin{itemize}
\item[-] for each $i\in T_{2}$, $\left \vert a_{i}\right \vert =q_{i}$;
\item[-] for each $i\in T_{1}$, $\left \vert a_{i}\right \vert \in \left \{
q_{i}-1,q_{i}\right \}  $;
\item[-] for some $i\in T_{1}$, $\left \vert a_{i}\right \vert =q_{i}$.
\end{itemize}
\end{theorem}

Let allocation $a$ be a Nash equilibrium. We need to designate the set of type
1 resources that are not assigned at $a$ up to their quotas: Let $L(a)=\{i\in
T_{1}\ |\ |a_{i}|=q_{i}-1\}$. Also, let $H(a)=M\setminus L(a)$. We refer to
the resources in $L(a)$ and in $H(a)$ as \textit{low} and \textit{high}
resources at $a$, respectively. The corollary below immediately follows from
the above theorem and it will be useful later on.

\begin{corollary}
\label{cor_NE_numberofLowHigh} Let allocation $a$ be a Nash equilibrium. Then,
$|L(a)|=\sum_{i\in M}q_{i}-n$ and $|H(a)|=m-|L(a)|$. Therefore, the number of low and high resources are the same at every Nash
equilibrium allocation.
\end{corollary}

The rest of this section is divided into two parts. We present our existence and non-existence results for laminar equilibrium notion in Section 3.1. We present our existence and non-existence results for contiguous and centralized equilibrium notions in Section 3.2.

\subsection{Existence and Non-Existence Results for Laminar Equilibrium}
\label{sec:LaminarNotExists}

In this section, we present our existence and non-existence results for laminar equilibrium. Our results resolves an open question in the literature. In their paper, Anshelevich et al. \cite{Ref_ACH2013} showed that in an RSG, for any given partition coalition structure, there exists a partition equilibrium, as stated in the following theorem.

\begin{theorem}[due to Anshelevich et al. \cite{Ref_ACH2013}] In an RSG, for any given
partition coalition structure $C\in C^{pe}$, there exists a Nash equilibrium
allocation which is $C$-stable. That is, in an RSG there always exists a partition equilibrium (which is also
Nash equilibrium).
\end{theorem}

They also conjectured that the following more general claim holds true: For any given laminar coalition structure, there exists a laminar equilibrium. We first prove their conjecture for the special setting where there are only two resources. This result is presented below, whose proof is given in Appendix B due to space limitations.

\begin{theorem} \label{teo_two_yesLE}
In a two-resource RSG, for any laminar coalition structure $C \in \mathcal{C}^{le}$, there exists a
$C$-stable allocation. That is, laminar equilibrium always exists in RSGs with two resources.
\end{theorem}

In the next section, we also prove that their conjecture holds for the special setting where the resources are identical (Corollary \ref{cor:iden_yesLE}), which is implied by the more general result that contiguous equilibrium always exists in RSGs with identical resources. Alas, we show that in the general setting, their conjecture does not hold.

The rest of this section is devoted to prove that a laminar equilibrium does not necessarily exist in RSGs (Theorem \ref{teo_noLE}). The example that we use to show Theorem \ref{teo_noLE} is an intricate one, consisting of a large number of agents and resources. We present it below.

\begin{example} \label{ex_noLE}
 Consider an RSG as follows:

\begin{itemize}
\item[-] There are $n=14052$ agents and $m=2001$ resources.

\item[-] Every resource is of type 1.

\item[-] The set of resources can be written as $M=M_{x}\cup M_{y}\cup M_{z}$
such that:

\begin{itemize}
\item[-] $M_{x}=\{x\}$ and $q_{x}=53$.

\item[-] $M_{y}=\{y_{1},y_{2},\ldots,y_{1000}\}$ where each resource in $M_{y}$ has the same cost function, and $q_{y}=8$ for all $y \in M_{y}$.

\item[-] $M_{z}=\{z_{1},z_{2},\ldots,z_{1000}\}$ where each resource in $M_{z}$ has the same cost function, and $q_{z}=7$ for all $z \in M_{z}$.

\item[-] For all $y \in M_{y}$ and $z \in M_{z}$, we have $\beta_{x}>\beta_{y}>\beta_{z}>$
$f_{x}\left(q_{x}-2\right)$. \hfill $\diamond$
\end{itemize}
\end{itemize}
\end{example}

\begin{theorem}
\label{teo_noLE}In an RSG, for $C\in\mathcal{C}^{le}$, it may be that no
allocation is $C$-stable. That is, in RSGs a laminar equilibrium is not guaranteed to exist.
\end{theorem}

\begin{proof}
In Example \ref{ex_noLE}, consider the following coalition structure:
$C=\left\{  c_{1},\ldots,c_{6}\right\}  \cup\mathcal{P}_{=1}(N)\cup\{N\}$ where the sets $c_{1},\ldots,c_{6}$ are disjoint and each has a cardinality
of $14052/6=2342$.

Note that $C$ is laminar. We prove the theorem by showing that no $C$-stable
allocation exists in Example \ref{ex_noLE}. By way of contradiction, suppose
that in Example \ref{ex_noLE} there exists an allocation $a$ which is $C$-stable.

Note that by Corollary \ref{cor_NE_numberofLowHigh}: $|L(a)|=1001$
($=53+8\times1000+7\times1000-14052$). And $|H(a)|=2001-1001=1000$. Since
$\mathcal{P}_{=1}(N)\subset C$, $a$ is a Nash equilibrium. Therefore, using
Corollary \ref{cor_NE_numberofLowHigh}, at allocation $a$ there are 1001 low
resources and 1000 high resources.

\bigskip

We divide the proof into six parts:

\bigskip

\textbf{(1)} We show that $x\in L\left(  a\right)  $.

\bigskip

By way of contradiction, suppose that $x\in H\left(  a\right)  $. Then, in
$M_{y}\cup M_{z}$, there are 1001 resources that are low. Let $i,i^{\prime}$
be two of them ($i\neq i^{\prime}$). Consider the agents $a_{i}\cup
a_{i^{\prime}}$. Note that $\left\vert a_{i}\cup a_{i^{\prime}}\right\vert
=q_{i}+q_{i^{\prime}}-2\leq14$. Let $N_{1}\cup N_{2}\subset a_{x}$ be such
that $N_{1}$ and $N_{2}$ are disjoint, $\left\vert N_{1}\right\vert =q_{i}$,
and $\left\vert N_{2}\right\vert =q_{i^{\prime}}$. We define allocation
$a^{\prime}$ from $a$ as follows.

\begin{itemize}
\item[-] Remove the agents in $N_{1}\cup N_{2}\cup a_{i}\cup a_{i^{\prime}}$ from their assigned resources.

\item[-] Assign agents in $N_{1}$ to resource $i$, assign agents in $N_{2}$ to
resource $i^{\prime}$, and assign agents in $a_{i}\cup a_{i^{\prime}}$ to resource
$x$.

\item[ ] (The assignments of remaining agents are the same as before.)
\end{itemize}

At allocation $a^{\prime}$, the agents assigned to resource $x$ are now better
off (since $x$ is now assigned $q_{x}-2$ agents). All other agents are equally
well-off at the two allocations. But then $a$ is not $N$-stable, a
contradiction. Thus, $x\in L\left(  a\right)  $.

\bigskip

\textbf{(2)} We show that $\left\vert H\left(  a\right)  \cap M_{y}\right\vert
\leq7$. (Hence, $\left\vert H\left(  a\right)  \cap M_{z}\right\vert \geq993$.)

\bigskip

By \textbf{(1)}, we know that $x\in L\left(  a\right)  $. Then, at $a$, in
$M_{y}\cup M_{z}$ there are 1000 high resources and 1000 low resources. By way
of contradiction, suppose that $\left\vert H\left(  a\right)  \cap
M_{y}\right\vert \geq8$. This implies that $\left\vert L\left(  a\right)  \cap
M_{z}\right\vert \geq8$. We define allocation $a^{\prime}$ from $a$ as
follows. We pick 7 high resources in $M_{y}$: Wlog., let $y_{1},\cdots
,y_{7}\in H\left(  a\right)  \cap M_{y}$. We pick 8 low resources in $M_{z}$:
Wlog., let $z_{1},\cdots,z_{8}\in L\left(  a\right)  \cap M_{z}$. We pick 49
agents assigned to $x$ at $a$: Let $N_{x}\subset a_{x}$ be such that
$\left\vert N_{x}\right\vert =49$. Then:

\begin{itemize}
\item[-] Remove the agents in $N_{x}\cup a_{y_{1}}\cup\cdots\cup a_{y_{7}}\cup
a_{z_{1}}\cup\cdots\cup a_{z_{8}}$ from their assigned resources.

\item[-] Assign the 49 agents in $N_{x}$ to resources $y_{1},\cdots,y_{7}$
such that each resource is assigned 7 agents.

\item[-] Assign the 56 agents in $a_{y_{1}}\cup\cdots\cup a_{y_{7}}$ to
resources $z_{1},\cdots,z_{8}$ such that each resource is assigned 7 agents.

\item[-] Assign the 48 agents in $a_{z_{1}}\cup\cdots\cup a_{z_{8}}$ to
resource $x$.

\item[ ] (The assignments of remaining agents are the same as before.)
\end{itemize}

At allocation $a^{\prime}$, the agents assigned to resource $x$ are now better
off (since $x$ is now assigned $q_{x}-2$ agents). The agents assigned to
resources $y_{1},\cdots,y_{7}$ are also better off (because they are now
assigned to low resources for which the beta value is smaller). The agents
assigned to resources $z_{1},\cdots,z_{8}$ are equally well-off (because they
are assigned to high resources at both $a$ and $a^{\prime}$). The agents
assigned to remaining resources are equally well-off. But then $a$ is not
$N$-stable, a contradiction. Thus, $\left\vert H\left(  a\right)  \cap
M_{y}\right\vert \leq7$. Hence, we also have $\left\vert H\left(  a\right)
\cap M_{z}\right\vert \geq993$.

\bigskip

\textbf{(3)} We show that there exists $c\in\left\{  c_{1},\ldots
,c_{6}\right\}  $ such that there are at least 1159 agents in $c$ which are
assigned to resources in $H\left(  a\right)  \cap M_{z}$ at allocation $a$.

\bigskip

Above, by \textbf{(2)}, at $a$ there are at least 993 high resources in
$M_{z}$. Since each of them is assigned 7 agents, at $a$ the number of agents
assigned to high resources in $M_{z}$ is at least $993\times7=6951$. But then
by the generalized pigeonhole principle, there is a coalition $c\in\left\{
c_{1},\ldots,c_{6}\right\}  $ such that at $a$ the number of agents in $c$
that are assigned to high resources in $M_{z}$ is at least $\left\lceil
\frac{6951}{6}\right\rceil =1159$.

\bigskip

\textbf{(4)} Let $c\in\left\{  c_{1},\ldots,c_{6}\right\}  $ be as described
in \textbf{(3)}. We show that there exists $z\in H\left(  a\right)  \cap
M_{z}$ such that there are at least two agents in $c$ that are assigned to $z$
at allocation $a$.

\bigskip

Note that at $a$ the number of high resources in $M_{z}$ is at most 1000
(because $\left\vert M_{z}\right\vert =1000$). By \textbf{(3)} we also know
that there are at least 1159 agents in $c$ which are assigned to high
resources in $M_{z}$ at allocation $a$. But then, by the pigeonhole principle,
there exists $z\in H\left(  a\right)  \cap M_{z}$ such that there are at least
two agents in $c$ that are assigned to $z$ at allocation $a$.

\bigskip

\textbf{(5)} Let $c\in\left\{  c_{1},\ldots,c_{6}\right\}  $ be as described
in \textbf{(3)}. We show that for each resource $y\in L\left(  a\right)
\cap\left(  M_{x}\cup M_{y}\right)  $, there are at least two agents in $c$
that are assigned to $y$ at allocation $a$.

\bigskip

By \textbf{(4)} there exists $z\in H\left(  a\right)  \cap M_{z}$ such that
there are at least two agents in $c$ that are assigned to $z$ at allocation
$a$. Thus, let $j,j^{\prime}\in c$ be such that $j\neq j^{\prime}$ and at $a$
the agents $j$ and $j^{\prime}$ are assigned to resource $z$.

By way of contradiction, suppose that there exists $y\in L\left(  a\right)
\cap\left(  M_{x}\cup M_{y}\right)  $ such that $\left\vert c\cap
a_{y}\right\vert \leq1$.

Suppose that $\left\vert c\cap a_{y}\right\vert =0$. We define allocation
$a^{\prime}$ from $a$ as follows: Agent $j$ is removed from resource $z$ and
then assigned to resource $y$. It is clear that at $a^{\prime}$ coalition $c$
is better off. But then $a$ is not $C$-stable, a contradiction. Thus,
$\left\vert c\cap a_{y}\right\vert \neq0$.

Suppose that $\left\vert c\cap a_{y}\right\vert =1$. Let $\widetilde{j}$ be
the agent in $c\cap a_{y}$. We define allocation $a^{\prime}$ from $a$ as
follows: Agents $j$ and $j^{\prime}$ are removed from resource $z$ and then
assigned to resource $y$, and agent $\widetilde{j}$ is removed from resource
$y$ and then assigned to resource $z$. Note that at $a^{\prime}$ the agents
$j$ and $j^{\prime}$ are equally well-off (they are still assigned to high
resources) and the agents in $c$ that are assigned to $z$ ($\widetilde{j}$ and
perhaps some other agents) are better off (because $z$ is now a low resource,
and the beta value for $z$ is smaller than the beta value for $y$). The
remaining agents in $c$ are equally well-off. But then $a$ is not $C$-stable,
a contradiction. Thus, $\left\vert c\cap a_{y}\right\vert \neq1$. Therefore, $\left\vert c\cap a_{y}\right\vert \geq2$.

\bigskip

\textbf{(6)} We conclude the proof as follows: Let $c$ be as described in
\textbf{(3)}. By \textbf{(1)} and \textbf{(2)}, there are at least 994
resources in $L(a)\cap\left(  M_{x}\cup M_{y}\right)  $. By \textbf{(5)}, the
number of agents in coalition $c$ that are assigned to resources in
$L(a)\cap\left(  M_{x}\cup M_{y}\right)  $ is at least $2\times994=1988$ at
allocation $a$. By \textbf{(3)}, there are at least 1159 agents in $c$ which
are assigned to resources in $H\left(  a\right)  \cap M_{z}$ at allocation
$a$. But then we get $\left\vert c\right\vert \geq1988+1159=3147$. This
contradicts the fact that $\left\vert c\right\vert =2342$.
\end{proof}

\begin{corollary}
\label{cor:NashParetoPartition}
In an RSG, there may not exist a Pareto efficient allocation $a$ that is both a partition equilibrium and a Nash equilibrium.
\end{corollary}

\begin{proof}
In proving Theorem \ref{teo_noLE}, we work with the coalition structure
$C=\left\{  c_{1},\ldots,c_{6}\right\}  \cup\mathcal{P}_{=1}(N)\cup\{N\}$. Here, the part $\left\{  c_{1},\ldots,c_{6}\right\}$
is a partition of the agents. The part $\mathcal{P}_{=1}(N)$ is the set of singleton coalitions. And the part $\{N\}$ is the grand coalition. If every RSG admitted a Pareto efficient allocation $a$ that is both a partition equilibrium and a Nash equilibrium, the RSG instance used in the proof of Theorem \ref{teo_noLE} would be $c$-stable for every coalition $c \in C$.
\end{proof}

In light of our findings in Section \ref{sec:equinot}, Theorem \ref{teo_noLE} also has the
following corollary:

\begin{corollary}
\label{cor:ContandCentEquDoesNotExist}
In an RSG, the existence of a contiguous equilibrium, or of a centralized equilibrium, is also not guaranteed.
\end{corollary}

%Feldman and Tennenholtz \cite{Ref_FT2009} showed that in RSGs a super strong equilibrium is not guaranteed to exist. Above, we considered less restrictive notions of equilibrium. Alas, as shown in Theorem \ref{teo_noLE}, an existence result cannot be obtained even for these less restrictive notions of equilibrium. Feldman and Tennenholtz \cite{Ref_FT2009} showed their negative result in an RSG where there are only two identical resources. A natural question is then whether or not existence results can be obtained for these less restrictive notions of equilibrium: (i) when resources are identical, and/or (ii) when the number of resources is two. We elaborate on this question in the rest of the paper.

\subsection{Existence and Non-Existence Results for Contiguous and Centralized Equilibrium}

In this section we present several existence and nonexistence results for contiguous and centralized equilibrium notions. We first prove that a contiguous equilibrium exists when the resources are identical.

\begin{theorem}
\label{teo_ident_yesCoE}
In an identical-resource RSG, for any given contiguous coalition structure $C\in\mathcal{C}^{coe}$, there exists a $C$-stable allocation. That is, in an RSG with identical resources, there always exists a contiguous equilibrium.
\end{theorem}

\begin{proof}
We proceed in two parts: First, in the identical resources setting we find a
sufficient condition for a Nash equilibrium allocation to be $c$-stable. Then,
using this condition, for $C\in\mathcal{C}^{coe}$, we construct a Nash
equilibrium allocation which is $C$-stable. Our construction is a simple one
and it helps us find a Nash equilibrium allocation which is $C$-stable in
linear time.

\bigskip

\textbf{(1)} Let $a$ be a Nash equilibrium allocation. Let $c$ be some
coalition of agents. We show that $a$ is $c$-stable if for every pair of
resources $i\in H(a)$ and $i^{\prime}\in L(a)$, $|a_{i}\cap c|\leq
|a_{i^{\prime}}\cap c|+1$.

\medskip

By way of contradiction, suppose that $c$ satisfies the above condition yet
$(c_{i})_{i=1}^{m}$ is a profitable deviation at $a$. Let $\overline{a}%
=a\circ(c_{i})_{i=1}^{m}$.

\medskip

Since resources are identical, their quotas and beta values are the same. Let
their quotas be $q$ and their beta values be $\beta$.

\medskip

Suppose that for resource $i$, $\left\vert \overline{a}_{i}\right\vert >q$.
Then, it is clear that $c_{i}\neq\emptyset$. Let $j\in c_{i}$. Then at
$\overline{a}$ agent $j$ incurs a cost greater than the minmaxcost $\alpha$.
Recall that at a Nash equilibrium allocation no agent incurs a cost greater
than $\alpha$. Then, $j$ is worse off at $\overline{a}$ than at $a$. This
contradicts that $(c_{i})_{i=1}^{m}$ is a profitable deviation at $a$. Thus,
for each resource $i$, we have $\left\vert \overline{a}_{i}\right\vert \leq
q_{i}$.

\medskip

Suppose that for each resource $i$, $\left\vert \overline{a}_{i}\right\vert
<q$. But then at $\overline{a}$, at each resource the cost incurred is less
than the minmaxcost $\alpha$, a contradiction. Hence, there exists a resource,
say $i$, such that $\left\vert \overline{a}_{i}\right\vert =q$. But then, the
theorem by Anshelevich et al. \cite{Ref_ACH2013} and its corollary (presented
above) show that: the allocation $\overline{a}$ is a Nash equilibrium;
$\left\vert H(a)\right\vert =\left\vert H(\overline{a})\right\vert $; and
$\left\vert L(a)\right\vert =\left\vert L(\overline{a})\right\vert $.

\medskip

Let $d_{X}\subseteq c$ denote the subset of coalition members that are
assigned to resources in $X\subseteq M$ at allocation $a$. Similarly, let
$\overline{d}_{X}\subseteq c$ denote the subset of coalition members that are
assigned to resources in $X\subseteq M$ at allocation $\overline{a}$.

\medskip

Since $(c_{i})_{i=1}^{m}$ is a profitable deviation at $a$, $\overline
{d}_{H(\overline{a})}\subseteq d_{H(a)}$. Also, note that if $\overline
{d}_{H(\overline{a})}=d_{H(a)}$, we also get that $\overline{d}_{L(\overline
{a})}=d_{L(a)}$. But then every coalition member is equally well-off at
allocations $a$ and $\overline{a}$. This contradicts that $(c_{i})_{i=1}^{m}$
is a profitable deviation at $a$. Hence, we must have $\overline
{d}_{H(\overline{a})}\subset d_{H(a)}$. Hence, $\left\vert d_{H(a)}\right\vert
-\left\vert \overline{d}_{H(\overline{a})}\right\vert >0$.

\medskip

Note that $\left\vert d_{H(a)}\right\vert =\left\vert d_{H(a)\cap
H(\overline{a})}\right\vert +\left\vert d_{H(a)\cap L(\overline{a}%
)}\right\vert $ and $\left\vert \overline{d}_{H(\overline{a})}\right\vert
=\left\vert \overline{d}_{H(a)\cap H(\overline{a})}\right\vert +\left\vert
\overline{d}_{L(a)\cap H(\overline{a})}\right\vert $. Also, note that for a
resource in $H(a)\cap H(\overline{a})$, the number of members of coalition $c$
assigned to it is the same at allocations $a$ and $\overline{a}$. Thus,
$\left\vert d_{H(a)\cap H(\overline{a})}\right\vert =\left\vert \overline
{d}_{H(a)\cap H(\overline{a})}\right\vert $. Thus, $\left\vert d_{H(a)}%
\right\vert -\left\vert \overline{d}_{H(\overline{a})}\right\vert =\left\vert
d_{H(a)\cap L(\overline{a})}\right\vert -\left\vert \overline{d}_{L(a)\cap
H(\overline{a})}\right\vert >0$. Hence, $\left\vert d_{H(a)\cap L(\overline
{a})}\right\vert >\left\vert \overline{d}_{L(a)\cap H(\overline{a}%
)}\right\vert $.

\medskip

Let $i^{\ast}\in L(a)$ be such that the value $|a_{i^{\ast}}\cap c|$ is
smallest. Let $|a_{i^{\ast}}\cap c|=s$. The condition above in \textbf{(1)}
implies that for each $i\in H(a)$, $|d_{i}|\leq s+1$. (Note that $d_{i}%
=a_{i}\cap c$.) Then, $\left\vert d_{H(a)\cap L(\overline{a})}\right\vert
\leq\left(  s+1\right)  \,\left\vert H(a)\cap L(\overline{a})\right\vert $.
Now consider a resource $i\in L(a)\cap H(\overline{a})$. By our choice of $s$
we know that $|d_{i}|\geq s$. (Note that $d_{i}=a_{i}\cap c$.) It is also
clear that $\left\vert \overline{d}_{i}\right\vert =\left\vert d_{i}%
\right\vert +1$ (because $i$ is a low resource at $a$ and it is a high
resource at $\overline{a}$). Then, $\left\vert \overline{d}_{i}\right\vert
\geq\left(  s+1\right)  $. Then, $\left\vert \overline{d}_{L(a)\cap
H(\overline{a})}\right\vert \geq\left(  s+1\right)  \,\left\vert L(a)\cap
H(\overline{a})\right\vert $. Above we also showed that $\left\vert
d_{H(a)\cap L(\overline{a})}\right\vert >\left\vert \overline{d}_{L(a)\cap
H(\overline{a})}\right\vert $. Then we get $\left\vert H(a)\cap L(\overline
{a})\right\vert >\left\vert L(a)\cap H(\overline{a})\right\vert $. Since
$\left\vert H(a)\right\vert =\left\vert H(a)\cap L(\overline{a})\right\vert
+\left\vert H(a)\cap H(\overline{a})\right\vert $ and $\left\vert
H(\overline{a})\right\vert =\left\vert L(a)\cap H(\overline{a})\right\vert
+\left\vert H(a)\cap H(\overline{a})\right\vert $, we obtain that
$H(a)>\left\vert H(\overline{a})\right\vert $. But above we showed that
$\left\vert H(a)\right\vert =\left\vert H(\overline{a})\right\vert $, a
contradiction. Therefore, allocation $a$ is $c$-stable.

\bigskip

\textbf{(2)} We show that for each $C\in\mathcal{C}^{coe}$, there exists an
allocation $a$ such that $a$ is a Nash equilibrium and the condition in
\textbf{(1)} is satisfied for every $c\in C$.

\medskip

Let $C\in\mathcal{C}^{coe}$. Let $P$ be a path in accordance with coalition
structure $C$. Wlog., let $P  = 1-2-\cdots-n$. To show \textbf{(2)} we will show that for $C$, there exists an allocation $a$ that satisfies the condition in \textbf{(1)}
for every $c\in C$. We will construct this allocation with Algorithm 1:

\begin{algorithm}
\begin{algorithmic}[1]
\State $j \leftarrow 1$
\For{$i = 1, \ldots, n$}
	\State \textbf{assign} \textit{agent} $i$ \textbf{to} \textit{resource} $j$
	\State $j \leftarrow (j \mod m) + 1 $
\EndFor
\end{algorithmic}
\caption{}
\end{algorithm}

Observe that Algorithm 1 places the agents one by one to the resources; and when it reaches the last resource, it rolls over to the first resource again. Therefore, the number of agents on each pair of resources will differ by at most one at the end of the algorithm. Since each resource's quota is the same, this means that the above algorithm constructs a Nash equilibrium allocation $a$ where $H(a) = \{1, \ldots, |H(a)|\}$ and $L(a) = \{|H(a)| + 1, \ldots, m\}$ due to Theorem \ref{thm:ExistenceNash}.

\medskip

Since for any pair of resources $i \in H(a)$ and $i' \in L(a)$ we have $i < i'$, before the algorithm assigns an agent to a low resource, it always assigns an agent to each high resource. Therefore, before the algorithm assigns a member of coalition $c \in C$ to a high resource $i \in H(a)$ for the first time, it might have been assigned another member of $c$ to a low resource $i' \in L(a)$ at most once.

\medskip

On the other hand, after the algorithm assigns a member of coalition $c \in C$ to a high resource $i \in H(a)$ for the first time, before it assigns another member of $c$ to $i$, it needs to assign a member of $c$ to all other resources, since the agents in $c$ are subsequently ordered under $P  = 1-2-\cdots-n$. Therefore, we have $|a_i \cap c| \leq |a_{i'} \cap c| + 1$ for any low resource $i' \in L(a)$.  This means that, at allocation $a$, for each coalition $c\in C$, the condition in \textbf{(1)} is satisfied.

%\medskip

%This completes our proof.\hfill
\end{proof}

In light of our findings in Section \ref{sec:equinot}, Theorem \ref{teo_ident_yesCoE} also has the
following corollary:

\begin{corollary}
\label{cor:iden_yesLE}
In an identical-resource RSG, for any given laminar coalition structure $C\in\mathcal{C}^{le}$, there exists a $C$-stable allocation. That is, in an RSG with identical resources, there always exists a laminar equilibrium.
\end{corollary}

Nonetheless, we show that contiguous equilibrium may not exist when there are two nonidentical resources. We present this result below, whose proof is given in Appendix C.

\begin{theorem}
\label{teo_two_noCOE}In an RSG with two resources, for $C\in\mathcal{C}^{coe}%
$, it may be that no allocation is $C$-stable. That is, in an RSG a contiguous equilibrium is not guaranteed to exist even
when the number of resources is restricted to 2.
\end{theorem}

We finally prove that a centralized equilibrium may not exist even for the two identical resources case. Note that what makes this result interesting is that even though centralized coalition structures contain $O(n^2)$ viable coalitions (instead of $2^n - 1$ coalitions), equilibrium may not exist in the two-identical resources setting, i.e., it strengths the non-existence result of super strong equilibrium given in Example \ref{ex_noSSE}. We present this result below, whose proof is given in Appendix C.

\begin{theorem}
\label{teo_twoident_noCEE}In an RSG with two identical resources, for
$C\in\mathcal{C}^{cee}$, it may be that no allocation is $C$-stable. That is, in an RSG a centralized equilibrium is not guaranteed to exist even
under the restriction that there are two identical resources.
\end{theorem}

\newpage

% Bibliography
\bibliographystyle{ACM-Reference-Format}
\bibliography{sample-bibliography}

\begin{thebibliography}{99}
\bibitem{Ref_ACH2013} Anshelevich, Elliot, Bugra Caskurlu, and Ameya
Hate. \textquotedblleft Partition equilibrium always exists in resource
selection games.\textquotedblright\ \emph{Theory of Computing Systems} 53.1
(2013): 73-85.

\bibitem{Ref_AKT2008} Ashlagi, Itai, Piotr Krysta, and Moshe
Tennenholtz. \textquotedblleft Social context games.\textquotedblright
\emph{International Workshop on Internet and Network Economics}. Springer,
Berlin, Heidelberg, 2008.

\bibitem{Ref_Aumann1959} Aumann, Robert J. \textquotedblleft Acceptable
points in general cooperative n-person games.\textquotedblright%
\ \emph{Contributions to the Theory of Games} (AM-40) 4 (1959): 287-324.

\bibitem{Ref_BPW1997} Bernheim, B. Douglas, Bezalel Peleg, and Michael
D. Whinston. \textquotedblleft Coalition-proof nash equilibria i.
concepts.\textquotedblright\ \emph{Journal of Economic Theory} 42.1 (1987): 1-12.

\bibitem{Ref_FT2009} Feldman, Michal, and Moshe Tennenholtz.
\textquotedblleft Structured coalitions in resource selection
games.\textquotedblright\ \emph{ACM Transactions on Intelligent Systems and
Technology (TIST)} 1, no. 1 (2010): 4.

\bibitem{Ref_HTW2006} Hayrapetyan, Ara, \'{E}va Tardos, and Tom
Wexler. \textquotedblleft The effect of collusion in congestion
games.\textquotedblright\ \emph{Proceedings of the thirty-eighth annual ACM
symposium on Theory of computing}. ACM, 2006.

\bibitem{Ref_HPPSV2011} Hoefer, Martin, Michal Penn, Maria Polukarov,
Alexander Skopalik, and Berthold V\"{o}cking. \textquotedblleft Considerate
equilibrium.\textquotedblright\ \emph{IJCAI}. 2011.

\bibitem{Ref_HL1997} Holzman, Ron, and Nissan Law-Yone.
\textquotedblleft Strong equilibrium in congestion games.\textquotedblright%
\ \emph{Games and economic behavior} 21.1-2 (1997): 85-101.

\bibitem{Ref_KLW1997} Konishi, Hideo, Michel Le Breton, and Shlomo
Weber. \textquotedblleft Equilibria in a model with partial
rivalry.\textquotedblright\ \emph{Journal of Economic Theory }72.1 (1997): 225-237.

\bibitem{Ref_KLW1999} Konishi, Hideo, Michel Le Breton, and Shlomo
Weber. \textquotedblleft On coalition-proof Nash equilibria in common agency
games.\textquotedblright\ \emph{Journal of Economic Theory} 85.1 (1999): 122-139.

\bibitem{Ref_KS2014} Kuniavsky, Sergey, and Rann Smorodinsky.
\textquotedblleft Equilibrium and potential in coalitional congestion
games.\textquotedblright\ \emph{Theory and decision} 76.1 (2014): 69-79.

\bibitem{Ref_M1979} Milinski, Manfred. \textquotedblleft An
evolutionarily stable feeding strategy in sticklebacks.\textquotedblright%
\ \emph{Zeitschrift f\"{u}r Tierpsychologie} 51.1 (1979): 36-40.

\bibitem{Ref_M1996} Milchtaich, Igal. \textquotedblleft Congestion
games with player-specific payoff functions.\textquotedblright\ \emph{Games
and economic behavior} 13.1 (1996): 111-124.

\bibitem{Ref_MS1996} Monderer, Dov, and Lloyd S. Shapley.
\textquotedblleft Potential games.\textquotedblright\ \emph{Games and economic
behavior} 14.1 (1996): 124-143.

\bibitem{Ref_MW1996} Moreno, Diego, and John Wooders.
\textquotedblleft Coalition-Proof Equilibrium.\textquotedblright\ \emph{Games
and Economic Behavior} 17.1 (1996): 80-112.

\bibitem{Ref_N1951} Nash, John. \textquotedblleft Non-cooperative
games.\textquotedblright\ \emph{Annals of mathematics} (1951): 286-295.

\bibitem{Ref_QM1994} Quint, Thomas, and Martin Shubik. \textquotedblleft A model of
migration.\textquotedblright Vol. 1088. \emph{Cowles Foundation for Research in Economics}, 1994.

\bibitem{Ref_R1973} Rosenthal, Robert W. \textquotedblleft A class of
games possessing pure-strategy Nash equilibria.\textquotedblright  \emph{International Journal of Game Theory} 2.1 (1973): 65-67.

\bibitem{esa} Umang Bhaskar, Phani Raj Lolakapuri. ``Equilibrium Computation in Atomic Splittable Routing Games.'' \emph{26th Annual European Symposium on Algorithms (ESA 2018)}: 1-14.

\bibitem{correa} J. R. Correa, A. S. Schulz, and N. E. Stier-Moses, ``A geometric approach to the price of anarchy in nonatomic congestion games,'' Games and Economic Behavior, vol. 64, no. 2, pp. 457–469, Nov. 2008.

\end{thebibliography}

\newpage

\appendix
\section*{Appendix A}

\begin{proof}[Proof of Theorem \ref{teo_mostrefined_SSE}]

We prove the theorem in six parts.

\bigskip

\textbf{(1)} $\mathcal{C}^{ne}\subseteq\mathcal{C}^{pe}$ since $\mathcal{P}_{=1}(N)$ is a partition over $N$. Let $C \in \mathcal{C}^{pe}$. Then for any two distinct coalitions $c_1, c_2 \in C$, we have $c_1 \cap c_2 = \emptyset$. Thus, $\mathcal{C}^{pe}\subseteq\mathcal{C}^{le}$. We now show that $\mathcal{C}^{le}\subseteq\mathcal{C}^{coe}$ by using strong
induction on the size of the set $N$.

\bigskip

The base case (when $|N|=1$) is trivial.
Thus, we move to the inductive step: Suppose that the relation $\mathcal{C}%
^{le}\subseteq\mathcal{C}^{coe}$ holds when $|N|\leq s$, where $s\geq1$. We
need to show that the relation $\mathcal{C}^{le}\subseteq\mathcal{C}^{coe}$
holds when $|N|=s+1$. Hence, suppose that $|N|=s+1$.

\bigskip

Consider an arbitrary $C\in\mathcal{C}^{le}$. We need to show that
$C\in\mathcal{C}^{coe}$. To do that, we need to find a path $P$ in accordance
with coalition structure $C$. If $C=\emptyset$ then any path works. Thus,
suppose that $C\neq\emptyset$. Also, note that for any path $P$ the agents in
$N$ are subsequently ordered. Thus, wlog. let $N\notin C$.

\bigskip

Since $C\neq\emptyset$ and $N\notin C$, there exists $c\in C$ such that
$c\subset N$ and for each $\widetilde{c}\in C$, $c\not \subset \widetilde{c}$.
Let $C^{1}=\{\widetilde{c}\in C|\widetilde{c}\subseteq c\}$ and $C^{2}%
=C\smallsetminus C^{1}$. Note that $|c|\leq s$, $|N\smallsetminus c|\leq s$,
and the coalition structures $C^{1}$ and $C^{2}$ are laminar. Then, by our
inductive hypothesis: There exist paths $P^{1}$ and $P^{2}$, whose sets of
vertices are $c$ and $N\smallsetminus c$ (in order), such that $P^{1}$ is in
accordance with $C^{1}$ and $P^{2}$ is in accordance with $C^{2}$. Now
consider the path $P:P^{1}-P^{2}$ (i.e., the agents in $c$ are ordered at the
beginning as in $P^{1}$, and then the agents in $N\smallsetminus c$ are
ordered at the end as in $P^{2}$). It is clear that $P$ is in accordance with
$C$. Thus, $C\in\mathcal{C}^{coe}$. Therefore, $\mathcal{C}^{le}%
\subseteq\mathcal{C}^{coe}$.
\bigskip

\textbf{(2)} We now show that $\mathcal{C}^{coe}\subseteq \mathcal{C}^{cee}$.

\bigskip

Let $C\in\mathcal{C}^{coe}$. All we need is to show that $C\in\mathcal{C}^{cee}$. Let $P$ be a
path in accordance with coalition structure $C$. Wlog., let $P=1-2-\cdots-n$.
To ease exposition, we will denote a coalition $c\in C$ by $[j^{\prime},j]$
where $j^{\prime}$ and $j$ are, in order, the smallest-index and the
largest-index agents in $c$.

\bigskip

To prove the desired result we first construct a list $L$ whose elements are
$\ast$'s and the agents in $N$. We begin with the empty ordered list $L=()$.
Then we expand $L$ by inserting in the list $\ast$'s and agents using the
following algorithm:

\bigskip

\textbf{For} $j = 1 \ldots n$:
\begin{enumerate}
\item Redefine $L$ by inserting $j$ at the end.

\item If there are no coalitions in $C$ of the form $\left[  j^{\prime
},j\right]  $ where $j^{\prime}\neq j$, proceed with the next iteration.

\item Otherwise, among coalitions of the form $\left[  j^{\prime},j\right]  $
where $j^{\prime}\neq j$, let the smallest value that $j^{\prime}$ takes be
$j^{\ast}$.

\item Let $r$ be the number of elements ($\ast$'s and agents) that come after
$j^{\ast}$ up to $j$ (including $j$).

\item Redefine $L$ by inserting $r$ consecutive $\ast$'s at the end, and then
proceed with the next iteration.
\end{enumerate}

\bigskip

We now show by induction that after the iteration for $j$, the list $L$ is
such that:

\begin{itemize}
\item[-] For each coalition $\left[  \overline{j},\widehat{j}\right]  \in C$
where $\overline{j}\neq\widehat{j}$ and $\widehat{j}\leq j$, the number of
consecutive $\ast$'s that succeed $\widehat{j}$ is at least as many as the
number of elements that succeed $\overline{j}$ up to agent $\widehat{j}$
(including $\widehat{j}$).
\end{itemize}

The base case, when $j=1$, is trivial: After the iteration for $j=1$ we obtain
that $L=(1)$ and the desired result holds vacuously. Hence, we move to the
inductive step.

\bigskip

Suppose that the above statement is true after the iteration for $j=s$ where
$s\geq1$. We need to show that the statement is true after the iteration for
$j=s+1$. Consider the iteration for $j=s+1$.

\bigskip

If there are no coalitions in $C$ of the form $\left[  j^{\prime},s+1\right]
$ where $j^{\prime}\neq s+1$, after the iteration the list $L$ remains
unchanged with the exception that agent $s+1$ has been inserted at the end.
The desired result then follows from the inductive hypothesis. (This is
because the part of the list that comes before agent $s+1$ did not change.)

\bigskip

Suppose that there are coalitions in $C$ of the form $\left[  j^{\prime
},s+1\right]  $ where $j^{\prime}\neq s+1$. Let $j^{\ast}$ and $r$ be as
defined in the algorithm. Then after the iteration the list $L$ remains
unchanged with the exception that agent $s+1$, followed by $r$ consecutive
$\ast$'s, are inserted at the end. For coalitions of the form $\left[
\overline{j},\widehat{j}\right]  \in C$ where $\widehat{j}<s+1$ the desired
result follows from the inductive hypothesis. For coalitions of the form
$\left[  \overline{j},s+1\right]  \in C$ the desired results follows by our
choices of $j^{\ast}$ and $r$. Thus, the list $L$ is as desired.

\bigskip

We now use the list $L$ produced by the above algorithm to show that
$C\in\mathcal{C}^{cee}$: Consider the following planar representation: In the
Cartesian space the agents $1,2,\cdots,n$ are subsequently positioned on a
line, and the distance between agents $j$ and $j+1$ is $1+$(the number of
$\ast$'s in the list $L$ between agents $j$ and $j+1$). For a coalition of the
form $\left[  j,j\right]  \in C$, draw a circle such that its radius is $1/2$,
and agent $j$ lies at the circle's center. For a coalition of the form
$[j^{\prime},j]\in C$ where $j^{\prime}\neq j$, draw a circle such that its
radius is the distance between agents $j^{\prime}$ and $j$, and agent
$j^{\prime}$ lies at the circle's center. By construction of $L$ this planar
representation is in accordance with coalition structure $C$. Therefore,
$C\in\mathcal{C}^{cee}$.

\bigskip

\textbf{(3)} We now show that $\mathcal{C}^{ne}\subset\mathcal{C}^{pe}\subset\mathcal{C}^{le}$ for $|N|\geq 2$. For
$|N|\geq2$, observe that:

\begin{itemize}
\item[(i)] $\left\{  \left\{  1,2\right\}  \right\}  \notin\mathcal{C}^{ne}$
but $\left\{  \left\{  1,2\right\}  \right\}  \in\mathcal{C}^{pe}$;

\item[(ii)] $\{\{1\},\{1,2\}\}\notin\mathcal{C}^{pe}$ but
$\{\{1\},\{1,2\}\}\in\mathcal{C}^{le}$.
\end{itemize}

\bigskip

Therefore, for $|N|\geq2$, $\mathcal{C}^{ne}\subset\mathcal{C}^{pe}%
\subset\mathcal{C}^{le}$.

\bigskip

\textbf{(4)} We now show that $\mathcal{C}^{le}\subset\mathcal{C}^{coe}$ for $|N|\geq 3$.

\bigskip

For $|N|\geq3$, observe that: $\{\{1,2\},\{2,3\}\}\notin\mathcal{C}^{le}$.
Note that the path $P:1-2-3$ is in accordance with this coalition structure.
Thus, $\{\{1,2\},\{2,3\}\}\in\mathcal{C}^{coe}$. Therefore, if $|N|\geq3$,
$\mathcal{C}^{le}\subset\mathcal{C}^{coe}$.

\bigskip

\textbf{(5)} We now show that $\mathcal{C}^{coe}\subset\mathcal{C}^{cee}$ for $|N|\geq 4$.

\bigskip

Let $|N|\geq4$. Let $C=\{\{1,2,3\},\{2,3,4\},\{3,4,1\},\{4,1,2\}\}$. To see
that $C\in\mathcal{C}^{cee}$, consider the following planar representation:
Agents $1,2,3,4$ lie at the four corners of a unit square. Each agent lies at
the center of a circle with radius 1. It is easy to verify that this planar
representation is in accordance with coalition structure $C$. Therefore,
$C\in\mathcal{C}^{cee}$.

\bigskip

Suppose that $C\in\mathcal{C}^{coe}$. Let $P$ be a path in accordance with
coalition structure $C$. Since $\{1,2,3\}\in C$ and $\{2,3,4\}\in C$, it must
be that under $P$ the agents $1,2,3,4$ are ordered next to one another. By
symmetry of $C$ w.r.t. agents we can assume wlog. that $P=\cdots
-1-2-3-4-\cdots$. But then under $P$ the agents in $\{3,4,1\}\in C$ are not
subsequently ordered, a contradiction. Thus, $C\notin\mathcal{C}^{coe}$.
Since $C\notin\mathcal{C}^{coe}$ but $C\in\mathcal{C}^{cee}$, we have
$\mathcal{C}^{coe}\subset\mathcal{C}^{cee}$ for $|N|\geq 4$.

\bigskip

\textbf{(6)} Notice that for each $C\in\mathcal{C}^{ne}\cup\mathcal{C}^{pe}\cup\mathcal{C}^{le}\cup\mathcal{C}^{coe}\cup\mathcal{C}^{cee}$, $C\subseteq\mathcal{P}_{\geq1}(N)$. This is because by definition any coalition structure $C$
is a subset of $\mathcal{P}_{\geq1}(N)$. We finally show that $\mathcal{P}_{\geq1}(N)\notin\mathcal{C}^{ne}\cup\mathcal{C}^{pe}\cup\mathcal{C}^{le}\cup\mathcal{C}^{coe}\cup\mathcal{C}^{cee}$ for $|N|\geq 3$. Suppose that $|N|\geq 3$. The observation that $\mathcal{P}%
_{\geq1}(N)\notin\mathcal{C}^{ne}\cup\mathcal{C}^{pe}\cup\mathcal{C}^{le}$ is
trivial and left for the reader.

\bigskip

Suppose that $\mathcal{P}_{\geq1}(N)\in\mathcal{C}^{coe}$. Then, there exists
a path $P$ in accordance with $\mathcal{P}_{\geq1}(N)$. By symmetry of
$\mathcal{P}_{\geq1}(N)$ w.r.t. agents we can assume wlog. that $P=1-2-\cdots
-n$. But then agents $1$ and $3$ are not subsequently ordered under $P$, a
contradiction. Thus, $\mathcal{P}_{\geq1}(N)\notin\mathcal{C}^{coe}$.

\bigskip

Suppose that $\mathcal{P}_{\geq1}(N)\in\mathcal{C}^{cee}$. Then, there exists
a planar representation $\left(  \phi,\psi\right)  $ in accordance with
$\mathcal{P}_{\geq1}(N)$. Under $\left(  \phi,\psi\right)  $, consider the
circles corresponding to coalitions $\{1,2\},\{2,3\},\{1,3\}$. It must be that
either a distinct agent lies at each of these three circles' center or there
is an agent that lies at the center of at least two circles. For this latter
case, wlog. suppose that agent $1$ lies at the center of the circles
corresponding to coalitions $\{1,2\}$ and $\{1,3\}$. Wlog., let $\psi
_{2}(\{1,2\})\geq\psi_{2}(\{1,3\})$. But then agent $3$ lies inside the circle
corresponding to coalition $\{1,2\}$, a contradiction. Thus, a distinct agent
lies at the center of each of these three circles.

\bigskip

Wlog., let agents $1,2,3$ lie at the centers of the circles corresponding to
coalitions $\{1,2\},\{2,3\},\{3,1\}$, respectively. For coalition $\{1,2\}$,
since $3$ lies outside of the circle corresponding to this coalition, we must
have $\psi_{2}(\{3,1\})>\psi_{2}(\{1,2\})$. Using similar arguments for
coalitions coalitions $\{2,3\}$ and $\{3,1\}$, we find that $\psi
_{2}(\{1,2\})>\psi_{2}(\{2,3\})$ and $\psi_{2}(\{2,3\})>\psi_{2}(\{3,1\})$.
But then we get $\psi_{2}(\{3,1\})>\psi_{2}(\{1,2\})>\psi_{2}(\{2,3\})>\psi
_{2}(\{3,1\})$, a contradiction. Therefore, $C\not \in \mathcal{C}^{cee}$.

\bigskip
This completes our proof.
\end{proof}

\section*{Appendix B}

We now prove that laminar equilibrium always exists in two-resource RSGs. We first present what we call \textquotedblleft the two-color theorem of
laminarity\textquotedblright (Theorem \ref{thm:2color}), the proof of which can be found in the Appendix. It lies at the heart of the proof of Theorem \ref{teo_two_yesLE}, the main result of this section. We believe that our two-color theorem may also be of
independent interest, in particular, in future studies on laminarity. In
simple terms it states that for any $N^{\prime}\subseteq N$ and any laminar
coalition structure $C$, the set $N^{\prime}$ can be partitioned into two
subsets of about equal size (i.e., into $N_{B}^{\prime}$ and $N_{W}^{\prime
}=N^{\prime}\smallsetminus N_{B}^{\prime}$ where $\left \vert \left \vert
N_{B}^{\prime}\right \vert -\left \vert N_{W}^{\prime}\right \vert \right \vert
\leq1$) such that for each $c\in C$, the set of coalition members in
$N^{\prime}$ (i.e., $N^{\prime}\cap c$) also becomes partitioned into two
subsets of about equal size (i.e., $\left \vert \left \vert N_{B}^{\prime}\cap
c\right \vert -\left \vert N_{W}^{\prime}\cap c\right \vert \right \vert \leq1$).
(The agents in $N_{B}^{\prime}$ and $N_{W}^{\prime}$ are referred to as
\textquotedblleft black agents\textquotedblright \ and \textquotedblleft white
agents,\textquotedblright \ respectively, and hence is the name of the theorem).

\begin{theorem}
\label{thm:2color}
\textbf{(the two-color theorem of laminarity)} Let $N^{\prime}\subseteq N$,
$\left \vert N^{\prime}\right \vert \geq1$. Let $k\geq1$ be such that
$\left \vert N^{\prime}\right \vert =2k-1$ or $\left \vert N^{\prime}\right \vert
=2k$. Then, for any laminar coalition structure $C$, the set $N^{\prime}$ can
be partitioned into the subsets $N_{B}^{\prime}$ and $N_{W}^{\prime}%
=N^{\prime}\smallsetminus N_{B}^{\prime}$ such that $\left \vert N_{B}^{\prime
}\right \vert =k$ and for each $c\in C$, $\left \vert \left \vert N_{B}^{\prime
}\cap c\right \vert -\left \vert N_{W}^{\prime}\cap c\right \vert \right \vert
\leq1$.
\end{theorem}

\begin{proof}

We first equip ourselves with some new terms and tools. Suppose that $C$ is such that $N\in C$ and $\mathcal{P}%
_{=1}(N)\subseteq C$. For $c,c^{\prime}\in C$, we say that \textquotedblleft%
$c$ is a child of $c^{\prime}$,\textquotedblright \ and \textquotedblleft%
$c^{\prime}$ is the mother of $c$,\textquotedblright \ if $c\subset c^{\prime}$
and there does not exist $c^{\prime \prime}\in C$ such that $c\subset
c^{\prime \prime}\subset c^{\prime}$. We recursively define the sets
$C^{1},C^{2},\cdots,C^{t}$ as follows:%
\begin{gather*}
C^{1}=\left \{  c\in C\left \vert c\text{ is not a child of any coalition
}c^{\prime}\in C\right.  \right \}  \text{,}\\
\text{and for }2\leq s\leq t\text{,}\\
C^{s}=\left \{  c\in C\left \vert c\text{ is a child of some coalition
}c^{\prime}\in C^{s-1}\right.  \right \}  \text{.}%
\end{gather*}

\bigskip

Above, $t$ is set such that $C^{t}$ is non-empty and no coalition in $C^{t}$
has a child in $C$. Note that $C^{1}=\left \{  N\right \}  $ and $C=C^{1}\cup
C^{2}\cup \cdots \cup C^{t}$.

\bigskip

For $c\in C$, let $\#(c)$ be the number of children of $c$. For $\#(c)>0$, we
label the children of $c$ as $c_{1},c_{2},\cdots,c_{\#(c)}$. To make precise
our labeling we use the following rule: among coalitions $c_{1},c_{2}%
,\cdots,c_{\#(c)}$, $c_{1}$ is the one that includes the smallest-index agent;
among coalitions $c_{2},\cdots,c_{\#(c)}$, $c_{2}$ is the one that includes
the smallest-index agent; and so on. Note that, since $\mathcal{P}%
_{=1}(N)\subseteq C$, $c=c_{1}\cup c_{2}\cup \cdots \cup c_{\#(c)} $. We are now
ready to proceed with our proof.

\bigskip

We can assume wlog that $N\in C$ and $\mathcal{P}_{=1}(N)\subseteq C$: If not,
we redefine the set $C$ as follows: $C:C\cup \left \{  N\right \}  \cup
\mathcal{P}_{=1}(N)$. Then, when we identify the sets $N_{B}^{\prime}$ and
$N_{W}^{\prime}$ such that the theorem's requirements are satisfied for
redefined $C$, clearly the theorem's requirements are also satisfied for $C$
before we redefined it. That is why the assumption that $N\in C$ and
$\mathcal{P}_{=1}(N)\subseteq C$ is innocuous.\bigskip

We prove the theorem using mathematical induction on $s$ as follows:
\bigskip

\begin{itemize}
\setlength\itemsep{1em}
\item[ ] \textbf{Base case}: There exist $N_{B}^{1}\subseteq N^{\prime}$ and
$N_{W}^{1}=N^{\prime}\smallsetminus N_{B}^{1}$ such that $\left \vert N_{B}%
^{1}\right \vert =k$, and for each $c\in C^{1}$, $\left \vert \left \vert
N_{B}^{1}\cap c\right \vert -\left \vert N_{W}^{1}\cap c\right \vert \right \vert
\leq1$.

\item[ ] \textbf{Inductive Step:}

\begin{itemize}
\item[ ] \emph{(inductive hypothesis)} For $s\in \left \{  1,2,3,\cdots
,t-1\right \}  $, let $N_{B}^{s}\subseteq N^{\prime}$ and $N_{W}^{s}=N^{\prime
}\smallsetminus N_{B}^{s}$ be such that $\left \vert N_{B}^{s}\right \vert =k$,
and for each $c\in C^{1}\cup C^{2}\cup \cdots \cup C^{s}$, $\left \vert
\left \vert N_{B}^{s}\cap c\right \vert -\left \vert N_{W}^{s}\cap c\right \vert
\right \vert \leq1$.

\item[ ] \emph{(inductive conclusion)} Then, there exist $N_{B}^{s+1}\subseteq
N^{\prime}$ and $N_{W}^{s+1}=N^{\prime}\smallsetminus N_{B}^{s+1}$ such that
$\left \vert N_{B}^{s+1}\right \vert =k$, and for each $c\in C^{1}\cup C^{2}%
\cup \cdots \cup C^{s+1}$, $\left \vert \left \vert N_{B}^{s+1}\cap c\right \vert
-\left \vert N_{W}^{s+1}\cap c\right \vert \right \vert \leq1$.
\end{itemize}
\end{itemize}

\bigskip

Note that when the proof by mathematical induction is done, the sets
$N_{B}^{\prime}=N_{B}^{t}$ and $N_{W}^{\prime}=N_{W}^{t}$ satisfy the
requirements in the theorem.

\bigskip

Showing the base case is trivial: Clearly, $C^{1}=\left \{  N\right \}  $. Then,
any two sets $N_{B}^{1}\subseteq N^{\prime}$ and $N_{W}^{1}=N^{\prime
}\smallsetminus N_{B}^{1}$, where $\left \vert N_{B}^{1}\right \vert =k$, will
be as required.

\bigskip

We now show the inductive step: Suppose the inductive hypothesis is true. Let
$N_{B}^{s+1}=N_{B}^{s}$ and $N_{W}^{s+1}=N_{W}^{s}$. If for each $c\in
C^{s+1}$, $\left \vert \left \vert N_{B}^{s+1}\cap c\right \vert -\left \vert
N_{W}^{s+1}\cap c\right \vert \right \vert \leq1$, we are done. Thus, suppose
that for some $c\in C^{s+1}$, $\left \vert \left \vert N_{B}^{s+1}\cap
c\right \vert -\left \vert N_{W}^{s+1}\cap c\right \vert \right \vert \geq2$.
Wlog, let $\left \vert N_{B}^{s+1}\cap c\right \vert -\left \vert N_{W}^{s+1}\cap
c\right \vert \geq2$. (The arguments are similar for the case when $\left \vert
N_{W}^{s+1}\cap c\right \vert -\left \vert N_{B}^{s+1}\cap c\right \vert \geq2$.)
Let $c^{\ast}\in C^{s} $ be the mother of $c$. By the inductive hypothesis,%
\[
-1\leq \left \vert N_{B}^{s+1}\cap c^{\ast}\right \vert -\left \vert N_{W}%
^{s+1}\cap c^{\ast}\right \vert =%
%TCIMACRO{\tsum \limits_{l=1}^{\#(c^{\ast})}}%
%BeginExpansion
{\textstyle \sum \limits_{l=1}^{\#(c^{\ast})}}
%EndExpansion
\left(  \left \vert N_{B}^{s+1}\cap c_{l}^{\ast}\right \vert -\left \vert
N_{W}^{s+1}\cap c_{l}^{\ast}\right \vert \right)  \leq1\text{.}%
\]

Since $c$ is a child of $c^{\ast}$ and $\left \vert N_{B}^{s+1}\cap
c\right \vert -\left \vert N_{W}^{s+1}\cap c\right \vert \geq2$, the above
inequality implies that there exists $l\in \left \{  1,\cdots,\#(c^{\ast
})\right \}  $ such that%
\[
\left \vert N_{B}^{s+1}\cap c_{l}^{\ast}\right \vert -\left \vert N_{W}^{s+1}\cap
c_{l}^{\ast}\right \vert \leq-1\text{.}%
\]

Thus, $\left \vert N_{W}^{s+1}\cap c_{l}^{\ast}\right \vert \geq \left \vert
N_{B}^{s+1}\cap c_{l}^{\ast}\right \vert +1\geq1$. From above, we also know
that $\left \vert N_{B}^{s+1}\cap c\right \vert \geq \left \vert N_{W}^{s+1}\cap
c\right \vert +2\geq2$. Thus, $N_{W}^{s+1}\cap c_{l}^{\ast}\neq \emptyset$ and
$N_{B}^{s+1}\cap c\neq \emptyset$. Let $j^{W}\in N_{W}^{s+1}\cap c_{l}^{\ast}$
and $j^{B}\in N_{B}^{s+1}\cap c$. We redefine the sets $N_{B}^{s+1}$ and
$N_{W}^{s+1}$ as follows:%
\[
N_{B}^{s+1}:N_{B}^{s+1}\smallsetminus \left \{  j^{B}\right \}  \cup \left \{
j^{W}\right \}  \text{ and }N_{W}^{s+1}:N_{W}^{s+1}\smallsetminus \left \{
j^{W}\right \}  \cup \left \{  j^{B}\right \}  \text{.}%
\]

Note that:

\begin{itemize}
\item[-] for $c_{l}^{\ast}$, after $N_{B}^{s+1}$ and $N_{W}^{s+1}$ are redefined,
the value $\left \vert \left \vert N_{B}^{s+1}\cap c_{l}^{\ast}\right \vert
-\left \vert N_{W}^{s+1}\cap c_{l}^{\ast}\right \vert \right \vert $ becomes
smaller or the same as before;

\item[-] for $c$, after $N_{B}^{s+1}$ and $N_{W}^{s+1}$ are redefined, the value
$\left \vert \left \vert N_{B}^{s+1}\cap c\right \vert -\left \vert N_{W}%
^{s+1}\cap c\right \vert \right \vert $ becomes smaller;

\item[-] for $\widetilde{c}\in C^{1}\cup C^{2}\cup \cdots \cup C^{s+1}%
\smallsetminus \left \{  c,c_{l}^{\ast}\right \}  $, after $N_{B}^{s+1} $ and
$N_{W}^{s+1}$ are redefined, the value $\left \vert \left \vert N_{B}^{s+1}%
\cap \widetilde{c}\right \vert -\left \vert N_{W}^{s+1}\cap \widetilde
{c}\right \vert \right \vert $ remains unchanged.
\end{itemize}

Obviously, the above process can be iterated and the sets $N_{B}^{s+1}$ and
$N_{W}^{s+1}$ can be redefined until the inductive conclusion is satisfied.
This concludes our proof.
\end{proof}

\bigskip
\bigskip

We next state Lemma \ref{lem:characterization}, which present a characterization of the sort of profitable deviations that may arise in a
two-resource RSG in a Nash equilibrium when $T_{1}=\left \{  1,2\right \}$. The proof of Lemma \ref{lem:characterization} is in the Appendix.

\begin{lemma}
\label{lem:characterization}
In a two-resource RSG, suppose that $T_{1}=\left \{  1,2\right \}  $ and
$T_{2}=\emptyset$. Let allocation $a$ be a Nash equilibrium such that for
resources $i$ and $i^{\prime}$, $\left \vert a_{i}\right \vert =q_{i}$ and
$\left \vert a_{i^{\prime}}\right \vert =q_{i^{\prime}}-1$. Then, for
$c\in \mathcal{P}_{\geq1}(N)$, $a$ is $c$-stable if and only if the conditions
C1, C2, and C3 below are satisfied:

\begin{itemize}
\item[ ] \textbf{C1}. if $\left \vert a_{i^{\prime}}\cap c\right \vert =0$ then
$\left \vert a_{i}\cap c\right \vert \leq1$;

\item[ ] \textbf{C2}. if $\beta_{i}=\beta_{i^{\prime}}$ and $\left \vert
a_{i^{\prime}}\cap c\right \vert >0$ then $\left \vert a_{i}\cap c\right \vert
\leq \left \vert a_{i^{\prime}}\cap c\right \vert +1$;

\item[ ] \textbf{C3}. if $\beta_{i}<\beta_{i^{\prime}}$ and $\left \vert
a_{i^{\prime}}\cap c\right \vert >0$ then $\left \vert a_{i}\cap c\right \vert
\leq \left \vert a_{i^{\prime}}\cap c\right \vert $.
\end{itemize}
\end{lemma}

\begin{proof}

Let $T_{1}=\left \{  1,2\right \}  $ and $T_{2}=\emptyset$. Let $a$ be a Nash
equilibrium such that $\left \vert a_{i}\right \vert =q_{i}$ and $\left \vert
a_{i^{\prime}}\right \vert =q_{i^{\prime}}-1$. We prove the two parts of the
biconditional statement separately.\bigskip

\textbf{(only if)}

\bigskip

By way of contradiction, suppose that $a$ is $c$-stable but one of the
conditions in the lemma is not satisfied.

\bigskip

If C1 is not satisfied, then $a_{i^{\prime}}\cap c=\emptyset$ and $\left \vert
a_{i}\cap c\right \vert \geq2$. Consider an agent $j\in a_{i}\cap c$. Note that
the set $\left(  a_{i}\cap c\right)  \smallsetminus \left \{  j\right \}  $ is
non-empty. Consider the deviation $\left(  c_{1},c_{2}\right)  $ such that
$c_{i^{\prime}}=\left \{  j\right \}  $ and $c_{i}=c\smallsetminus \left \{
j\right \}  $. It is clear that $\left(  c_{1},c_{2}\right)  $ is a profitable
deviation by coalition $c$ at $a$, a contradiction. Therefore, if $a$ is
$c$-stable the condition C1 is satisfied.

\bigskip

If C2 is not satisfied, then $\beta_{i}=\beta_{i^{\prime}}$, $\left \vert
a_{i^{\prime}}\cap c\right \vert >0$, and $\left \vert a_{i}\cap c\right \vert
\geq \left \vert a_{i^{\prime}}\cap c\right \vert +2$. Let $\left \vert
a_{i^{\prime}}\cap c\right \vert =k$ and $\left \vert a_{i}\cap c\right \vert
=k+2+s$, where $k>0$ and $s\geq0$. Let $\widetilde{c}\subset a_{i}\cap c$ be
such that $\left \vert \widetilde{c}\right \vert =k+1$. Consider the deviation
$\left(  c_{1},c_{2}\right)  $ such that $c_{i^{\prime}}=\widetilde{c}$ and
$c_{i}=c\smallsetminus \widetilde{c}$. It is clear that $\left(  c_{1}%
,c_{2}\right)  $ is a profitable deviation by coalition $c$ at $a$, a
contradiction. Therefore, if $a$ is $c$-stable the condition C2 is satisfied.

\bigskip

If C3 is not satisfied, then $\beta_{i}<\beta_{i^{\prime}}$, $\left \vert
a_{i^{\prime}}\cap c\right \vert >0$, and $\left \vert a_{i}\cap c\right \vert
>\left \vert a_{i^{\prime}}\cap c\right \vert $. Let $\left \vert a_{i^{\prime}%
}\cap c\right \vert =k$ and $\left \vert a_{i}\cap c\right \vert =k+1+s$, where
$k>0$ and $s\geq0$. Let $\widetilde{c}\subseteq a_{i}\cap c$ be such that
$\left \vert \widetilde{c}\right \vert =k+1$. Consider the deviation $\left(
c_{1},c_{2}\right)  $ such that $c_{i^{\prime}}=\widetilde{c}$ and
$c_{i}=c\smallsetminus \widetilde{c}$. It is clear that $\left(  c_{1}%
,c_{2}\right)  $ is a profitable deviation by coalition $c$ at $a$, a
contradiction. Therefore, if $a$ is $c$-stable the condition C3 is satisfied.

\bigskip

Therefore, if $a$ is $c$-stable, then the conditions given in the lemma are
all satisfied.

\bigskip

\textbf{(if)}

\bigskip

By way of contradiction, suppose that for allocation $a$ the conditions C1,
C2, C3 are satisfied but $a$ is not $c$-stable. Then there exists a profitable
deviation $\left(  c_{1},c_{2}\right)  $ at allocation $a$. Let $a^{\prime
}=a\circ \left(  c_{1},c_{2}\right)  $.\bigskip

Suppose that for some resource $i^{\prime \prime}\in \left \{  i,i^{\prime
}\right \}  $, $\left \vert a_{i^{\prime \prime}}^{\prime}\right \vert
>q_{i^{\prime \prime}}$. Then $\left \vert a_{i^{\prime \prime}}^{\prime
}\right \vert >\left \vert a_{i^{\prime \prime}}\right \vert $. Since $\left \vert
a_{i^{\prime \prime}}^{\prime}\right \vert =\left \vert a_{i^{\prime \prime}%
}\right \vert -\left \vert a_{i^{\prime \prime}}\cap c\right \vert +\left \vert
c_{i^{\prime \prime}}\right \vert $, we must have $\left \vert c_{i^{\prime
\prime}}\right \vert >\left \vert a_{i^{\prime \prime}}\cap c\right \vert $. Then
$c_{i^{\prime \prime}}\neq \emptyset$. Let $j\in c_{i^{\prime \prime}}$. Since
$\left \vert a_{i^{\prime \prime}}^{\prime}\right \vert >q_{i}$, the cost that
$j$ incurs at $a^{\prime}$ is greater than $u$. But at $a$ the cost that $j$
incurs is less than or equal to $u$ (because $a$ is a Nash equilibrium; see
Proposition 1). This contradicts that $\left(  c_{1},c_{2}\right)  $ is a
profitable deviation at $a $. Thus, it must be that for each resource
$i^{\prime \prime}\in \left \{  i,i^{\prime}\right \}  $, $a_{i^{\prime \prime}%
}^{\prime}\leq q_{i^{\prime \prime}}$.\bigskip

Note that $\left \vert N\right \vert =\left \vert a_{i}\right \vert +\left \vert
a_{i^{\prime}}\right \vert =q_{i}+q_{i^{\prime}}-1$. Then the fact that for
each resource $i^{\prime \prime}\in \left \{  i,i^{\prime}\right \}  $,
$a_{i^{\prime \prime}}^{\prime}\leq q_{i^{\prime \prime}}$, implies that either
$\left \vert a_{i}^{\prime}\right \vert =q_{i}$ and $\left \vert a_{i^{\prime}%
}^{\prime}\right \vert =q_{i^{\prime}}-1$, or $\left \vert a_{i}^{\prime
}\right \vert =q_{i}-1$ and $\left \vert a_{i^{\prime}}^{\prime}\right \vert
=q_{i^{\prime}}$.

\bigskip

Suppose that $\left \vert a_{i}^{\prime}\right \vert =q_{i}$ and $\left \vert
a_{i^{\prime}}^{\prime}\right \vert =q_{i^{\prime}}-1$. Then it is clear that
the sum of the costs incurred by members of coalition $c$ is the same at
allocations $a$ and $a^{\prime}$. But then, if at $a^{\prime}$ an agent in $c$
is better off (compared to at $a$), it must be that another agent in $c$ is
worse off at $a^{\prime}$. But then $\left(  c_{1},c_{2}\right)  $ cannot be a
profitable deviation at $a$, a contradiction. Therefore, we must have
$\left \vert a_{i}^{\prime}\right \vert =q_{i}-1$ and $\left \vert a_{i^{\prime}%
}^{\prime}\right \vert =q_{i^{\prime}}$.

\bigskip

If $a_{i^{\prime}}\cap c=\emptyset$, then by C1 we get $\left \vert a_{i}\cap
c\right \vert \leq1$. Since $c$ is non-empty, we get $\left \vert a_{i}\cap
c\right \vert =1$ and $\left \vert c\right \vert =1$. But then $\left(
c_{1},c_{2}\right)  $ is a profitable deviation by a single-agent coalition,
contradicting that $a$ is a Nash equilibrium. Therefore, we obtain that
$a_{i^{\prime}}\cap c\neq \emptyset$. Let $\left \vert a_{i^{\prime}}\cap
c\right \vert =k$ where $k>0$.\bigskip

Suppose that $\beta_{i}>\beta_{i^{\prime}}$. Consider an agent $j\in
a_{i^{\prime}}\cap c$. Note that the cost that $j$ incurs at $a$ is
$\beta_{i^{\prime}}$, and the cost that $j$ incurs at $a^{\prime}$ is either
$\beta_{i}$ or $u$. Either way $j$ is worse off at allocation $a^{\prime}$,
contradicting that $\left(  c_{1},c_{2}\right)  $ is a profitable deviation at
$a$. Therefore, $\beta_{i}\leq \beta_{i^{\prime}}$.\bigskip

Suppose that $\beta_{i}=\beta_{i^{\prime}}$. Then, by C2, we find that
$\left \vert a_{i}\cap c\right \vert =s$ where $s\leq k+1$. Since $\left \vert
a_{i}^{\prime}\right \vert =q_{i}-1$, $\left \vert a_{i}\right \vert =q_{i}$, and
$\left \vert a_{i}^{\prime}\right \vert =\left \vert a_{i}\right \vert -\left \vert
a_{i}\cap c\right \vert +\left \vert a_{i}^{\prime}\cap c\right \vert $, we
obtain that $\left \vert a_{i}^{\prime}\cap c\right \vert =\left \vert a_{i}\cap
c\right \vert -1$. Then $\left \vert a_{i}^{\prime}\cap c\right \vert =s-1$. At
$a$ the agents in $a_{i^{\prime}}\cap c$ incur a cost equal to $\beta
_{i^{\prime}}<u$. Hence, at $a^{\prime}$ they cannot be assigned to resource
$i^{\prime}$ (where the cost incurred is $u$). Therefore, $a_{i^{\prime}}\cap
c\subseteq a_{i}^{\prime}\cap c$. Then $\left \vert a_{i}^{\prime}\cap
c\right \vert \geq \left \vert a_{i^{\prime}}\cap c\right \vert $. Therefore,
$s-1\geq k$. Since we also know that $s\leq k+1$, we obtain that $s=k+1$. Then
$a_{i^{\prime}}\cap c=a_{i}^{\prime}\cap c$. This means that at $a^{\prime} $,
agents in $a_{i^{\prime}}\cap c$ are assigned to resource $i$ and incur a cost
equal to $\beta_{i}$, and agents in $a_{i}\cap c$ are assigned to resource
$i^{\prime}$ and incur a cost equal to $u$. But then all agents are equally
well off at $a^{\prime}$ and $a$, contradicting that $\left(  c_{1}%
,c_{2}\right)  $ is a profitable deviation at $a$. Therefore, $\beta_{i}%
\neq \beta_{i^{\prime}}$.\bigskip

Suppose that $\beta_{i}<\beta_{i^{\prime}}$. Then, by C3, we find that
$\left \vert a_{i}\cap c\right \vert =s$ where $s\leq k$. Since $\left \vert
a_{i}^{\prime}\right \vert =q_{i}-1$, $\left \vert a_{i}\right \vert =q_{i}$, and
$\left \vert a_{i}^{\prime}\right \vert =\left \vert a_{i}\right \vert -\left \vert
a_{i}\cap c\right \vert +\left \vert a_{i}^{\prime}\cap c\right \vert $, we
obtain that $\left \vert a_{i}^{\prime}\cap c\right \vert =\left \vert a_{i}\cap
c\right \vert -1$. Then $\left \vert a_{i}^{\prime}\cap c\right \vert =s-1$. Note
that at $a$ the agents in $a_{i^{\prime}}\cap c$ incur a cost equal to
$\beta_{i^{\prime}}<u$, and hence at $a^{\prime}$ they cannot be assigned to
resource $i^{\prime}$ (where the cost incurred is $u$). Therefore,
$a_{i^{\prime}}\cap c\subseteq a_{i}^{\prime}\cap c$. Then $\left \vert
a_{i}^{\prime}\cap c\right \vert \geq \left \vert a_{i^{\prime}}\cap c\right \vert
$. Therefore, $s-1\geq k$. But this contradicts with the fact that $s\leq k$.
Therefore, $\beta_{i}<\beta_{i^{\prime}}$ cannot be true.\bigskip

Since our supposition that $\left(  c_{1},c_{2}\right)  $ is a profitable
deviation at $a$ always leads to a contradiction, we find that when the
conditions C1, C2, C3 are satisfied, the allocation $a$ is $c$-stable.\bigskip

This concludes our proof.
\end{proof}

Before presenting and proving the main result of this section, we will introduce
some new tools.

\medskip

The $\gamma$\emph{-value} of an allocation $a$ w.r.t. a coalition structure
$C$, to be denoted by $\gamma(a,C)$, is defined as follows:%
\begin{align*}
\gamma(a,C)  &  =%
%TCIMACRO{\tsum _{i\in M}}%
%BeginExpansion
{\textstyle \sum_{i\in M}}
%EndExpansion%
%TCIMACRO{\tsum \limits_{c\in C}}%
%BeginExpansion
{\textstyle \sum \limits_{c\in C}}
%EndExpansion
1(c,a_{i})\text{, where}\\
1(c,a_{i})  &  :\left \{
\begin{array}
[c]{l}%
1\text{ if }c\cap a_{i}\neq \emptyset \\
0\text{ otherwise}%
\end{array}
\right.  \text{.}%
\end{align*}

Loosely speaking, the $\gamma$\emph{-value} of allocation $a$ is a cumulative
measure of how \textquotedblleft widely\textquotedblright \ coalitions are
spread to resources at allocation $a$.

\medskip

The $\beta$\emph{-value} of an allocation $a$, to be denoted by $\beta(a)$, is
defined as follows:%
\[
\beta(a)=%
%TCIMACRO{\tsum \limits_{i\in M}}%
%BeginExpansion
{\textstyle \sum \limits_{i\in M}}
%EndExpansion
f_{i}(\left \vert a_{i}\right \vert )\text{.}%
\]

That is, the $\beta$\emph{-value} of allocation $a$ is the sum of the costs at resources at allocation $a$.

\medskip

We say that allocation $a^{\prime}$ $\gamma \beta$-dominates allocation $a$
w.r.t. $C$ if $\gamma(a^{\prime},C)>\gamma(a,C)$ or if $\gamma(a^{\prime
},C)=\gamma(a,C)$ and $\beta(a^{\prime})<\beta(a)$.

\medskip

Let $A\subseteq \mathcal{A}$ be a subset of allocations. Clearly, there exists
$a\in A$ such that, for each $a^{\prime}\in A\smallsetminus \left \{  a\right \}
$, either $a$ $\gamma \beta$-dominates $a^{\prime}$ w.r.t. $C$, or $a$ and
$a^{\prime}$ cannot be compared according to the $\gamma \beta$-domination
relation w.r.t. $C$. We refer to such an allocation $a$ as a \textquotedblleft
maximal element in $A$ according to the $\gamma \beta$-domination relation
w.r.t. $C$.\textquotedblright \ Note that there may be more than one maximal
elements in $A$.

\bigskip

We are now ready to present the main result of this section.

\begin{proof}[Proof of Theorem \ref{teo_two_yesLE}]
We show the existence of a $C$-stable allocation separately for the following
three cases:

\begin{itemize}
\item[-] Case 1: $T_{2}\neq \emptyset$.

\item[-] Case 2: $T_{1}=\left \{  1,2\right \}  $ and $\beta_{1}=\beta_{2}$.

\item[-] Case 3: $T_{1}=\left \{  1,2\right \}  $ and $\beta_{1}\neq \beta_{2}$.
\end{itemize}

Let allocation $a$ be a Nash equilibrium. (Its existence is by Theorem \ref{thm:ExistenceNash}.)

\bigskip

\textbf{Case 1:} $T_{2}\neq \emptyset$.

\bigskip

By Theorem \ref{thm:ExistenceNash}, $\left \vert T_{1}\right \vert =\left \vert T_{2}\right \vert
=1$; and $\left \vert a_{1}\right \vert =q_{1}$ and $\left \vert a_{2}\right \vert
=q_{2}$. But then, it is trivial to see that $a$ is a super strong equilibrium and
hence it is $C$-stable.

\bigskip

\textbf{Case 2:} $T_{1}=\left \{  1,2\right \}  $ and $\beta_{1}=\beta_{2}$.

\bigskip

By Theorem \ref{thm:ExistenceNash}, either $\left \vert a_{1}\right \vert =q_{1}-1$ and
$\left \vert a_{2}\right \vert =q_{2}$ or $\left \vert a_{1}\right \vert =q_{1}$
and $\left \vert a_{2}\right \vert =q_{2}-1$. Thus, $\left \vert N\right \vert
=q_{1}+q_{2}-1$.

\medskip

Let $k\geq1$ be such that $\left \vert N\right \vert =2k-1$ or $\left \vert
N\right \vert =2k$. Wlog, let $q_{1}\leq q_{2}$. Then, $q_{1}\leq k\leq q_{2} $.

\medskip

By Theorem \ref{thm:2color}, there exist $N_{B}^{\prime}\subseteq
N$ and $N_{W}^{\prime}=N\smallsetminus N_{B}^{\prime}$ such that $\left \vert
N_{B}^{\prime}\right \vert =k$ and for each $c\in C$, $\left \vert N_{B}%
^{\prime}\cap c\right \vert \leq \left \vert N_{W}^{\prime}\cap c\right \vert +1$.

\medskip

Let $\widetilde{N}_{B}\subseteq N_{B}^{\prime}$ be such that $\left \vert
\widetilde{N}_{B}\right \vert =q_{1}$. Let $\widetilde{N}_{W}=N\smallsetminus
\widetilde{N}_{B}$. Note that $\left \vert \widetilde{N}_{W}\right \vert
=q_{2}-1$. Since $\widetilde{N}_{B}\subseteq N_{B}^{\prime} $ and
$N_{W}^{\prime}\subseteq \widetilde{N}_{W}$, we obtain that for each $c\in C$,
$\left \vert \widetilde{N}_{B}\cap c\right \vert \leq \left \vert \widetilde
{N}_{W}\cap c\right \vert +1$. Therefore, for allocation $a^{\prime}$ such that
$a_{1}^{\prime}=\widetilde{N}_{B}$ and $a_{2}^{\prime}=\widetilde{N}_{W}$, the
conditions C1 and C2 in Lemma \ref{lem:characterization} are satisfied while the condition C3 is not
applicable. (To ease comparison with lemma conditions, note that $i $ and
$i^{\prime}$ in the lemma statement are 1 and 2 in here, in order.) Therefore,
by Lemma \ref{lem:characterization}, $a^{\prime}$ is $C$-stable.

\bigskip

\textbf{Case 3:} $T_{1}=\left \{  1,2\right \}  $ and $\beta_{1}\neq \beta_{2}$.

\bigskip

If $a$ is $C$-stable, we are done. If not, we proceed as follows: We show the
existence of an allocation $a^{\prime}$ such that $a^{\prime}$ is a Nash
equilibrium and $a^{\prime}$ $\gamma \beta$-dominates $a$ w.r.t. $C$. This
proves that a $C$-stable allocation exists because: If $a^{\prime}$ turns out
to be $C$-stable, we are done. Otherwise, we can iterate the same arguments:
We can find an allocation $a^{\prime \prime}$ such that $a^{\prime \prime}$ is a
Nash equilibrium and $a^{\prime \prime}$ $\gamma \beta$-dominates $a^{\prime}$
w.r.t. $C$, and so on. Since there exists a maximal element in the set of Nash
equilibria according to the $\gamma \beta$-domination relation w.r.t. $C$, our
iterations must eventually yield a $C$-stable allocation.

\medskip

Therefore, suppose that $a$ is not $C$-stable. Let $i,i^{\prime}\in \left \{
1,2\right \}  $ be such that $\left \vert a_{i}\right \vert =q_{i}$ and
$\left \vert a_{i^{\prime}}\right \vert =q_{i^{\prime}}-1$.

\medskip

By Lemma \ref{lem:characterization}, there exists $c\in C$ such that one of the conditions C1, C2, and
C3 in Lemma \ref{lem:characterization} is not satisfied. Since $\beta_{1}\neq \beta_{2}$, C2 is not
applicable. Thus, either C1 or C3 is not satisfied.

\medskip

Suppose that the condition C1 is not satisfied. Then, there exists $c\in C$
such that $\left \vert a_{i^{\prime}}\cap c\right \vert =0$ and $\left \vert
a_{i}\cap c\right \vert \geq2$. Let $j,j^{\prime}\in a_{i}\cap c$, $j\neq
j^{\prime}$. Let allocation $a^{\prime}$ be such that $a_{i}^{\prime}%
=a_{i}\smallsetminus \left \{  j^{\prime}\right \}  $ and $a_{i^{\prime}}%
^{\prime}=a_{i^{\prime}}\cup \left \{  j^{\prime}\right \}  $; hence, $j\in
a_{i}^{\prime}\cap c$, $j^{\prime}\in a_{i^{\prime}}^{\prime}\cap c$, and
$1\left(  c,a_{i}^{\prime}\right)  =1\left(  c,a_{i^{\prime}}^{\prime}\right)
=1$. By Theorem \ref{thm:ExistenceNash}, $a^{\prime}$ is a Nash equilibrium. Also, note that
$\gamma \left(  a^{\prime},C\right)  >\gamma \left(  a,C\right)  $ because:

\begin{itemize}
\setlength\itemsep{1em}
\item[-] For each $c^{\prime}\in C$ such that $c^{\prime}\cap c=\emptyset$,%
\[
1\left(  c^{\prime},a_{i}\right)  +1\left(  c^{\prime},a_{i^{\prime}}\right)
=1\left(  c^{\prime},a_{i}^{\prime}\right)  +1\left(  c^{\prime},a_{i^{\prime
}}^{\prime}\right)  \text{.}%
\]

(Because agents in $N\smallsetminus c$ are allocated to resources in exactly
the same way at allocations $a$ and $a^{\prime}$.)

\item[-] For each $c^{\prime}\in C$ such that $c\subset c^{\prime}$,%
\[
1\left(  c^{\prime},a_{i}^{\prime}\right)  +1\left(  c^{\prime},a_{i^{\prime}%
}^{\prime}\right)  \geq1\left(  c^{\prime},a_{i}\right)  +1\left(  c^{\prime
},a_{i^{\prime}}\right)  .
\]

(Because $1\left(  c,a_{i}^{\prime}\right)  +1\left(  c,a_{i^{\prime}}%
^{\prime}\right)  =2$ and hence $1\left(  c^{\prime},a_{i}^{\prime}\right)
+1\left(  c^{\prime},a_{i^{\prime}}^{\prime}\right)  =2$.)

\item[-] For $c$,%
\[
1\left(  c,a_{i}^{\prime}\right)  +1\left(  c,a_{i^{\prime}}^{\prime}\right)
>1\left(  c,a_{i}\right)  +1\left(  c,a_{i^{\prime}}\right)  \text{.}%
\]

(Because $1\left(  c,a_{i}^{\prime}\right)  +1\left(  c,a_{i^{\prime}}%
^{\prime}\right)  =2$ and $1\left(  c,a_{i}\right)  +$ $1\left(
c,a_{i^{\prime}}\right)  =1$.)

\item[-] For each $c^{\prime}\in C$ such that $c^{\prime}\subset c$,%
\[
1\left(  c^{\prime},a_{i}^{\prime}\right)  +1\left(  c^{\prime},a_{i^{\prime}%
}^{\prime}\right)  \geq1\left(  c^{\prime},a_{i}\right)  +1\left(  c^{\prime
},a_{i^{\prime}}\right)  \text{.}%
\]

(Because $1\left(  c,a_{i}\right)  +1\left(  c,a_{i^{\prime}}\right)  =1$, and
hence, $1\left(  c^{\prime},a_{i}\right)  +1\left(  c^{\prime},a_{i^{\prime}%
}\right)  =1$.)
\end{itemize}

Thus, as required, allocation $a^{\prime}$ is a Nash equilibrium and
$a^{\prime}$ $\gamma \beta$-dominates $a$ w.r.t. $C$.

\bigskip

Suppose that the condition C3 is not satisfied. Thus, $\beta_{i}%
<\beta_{i^{\prime}}$ and there exists $c\in C$ such that $\left \vert a_{i}\cap
c\right \vert >\left \vert a_{i^{\prime}}\cap c\right \vert >0$. Let $k$ be such
that $\left \vert a_{i^{\prime}}\cap c\right \vert =k-1$. Note that $\left \vert
a_{i}\cap c\right \vert \geq k\geq2$.

\medskip

For each $j\in a_{i^{\prime}}\cap c$, we define agent $\widetilde{j}$ as
follows: Let $\widetilde{c}\in C$ be such that $\widetilde{c}\subseteq c$,
$a_{i}\cap \widetilde{c}\neq \emptyset$, and there does not exist $\overline
{c}\in C$ such that $\overline{c}\subset \widetilde{c}$ and $a_{i}\cap
\overline{c}\neq \emptyset$. Let $\widetilde{j}$ be the smallest-index agent in
$a_{i}\cap \widetilde{c}$.

\medskip

Let $S_{i}=\left \{  \widetilde{j}\left \vert j\in a_{i^{\prime}}\cap
c\right \vert \right \}  $. Note that, since $\left \vert a_{i^{\prime}}\cap
c\right \vert =k-1$, $\left \vert S_{i}\right \vert \leq k-1$. Let $\overline
{S}_{i}$ be such that $S_{i}\subset \overline{S}_{i}\subseteq \left(  a_{i}\cap
c\right)  $ and $\left \vert \overline{S}_{i}\right \vert =k$. Let $\overline
{S}_{i,i^{\prime}}=\overline{S}_{i}\cup \left(  a_{i^{\prime}}\cap c\right)  $. Note that:

\begin{itemize}
\item[-] $\left \vert \overline{S}_{i,i^{\prime}}\right \vert =2k-1\geq3$.

\item[-] and for each $c^{\prime}\in C$ such that $c^{\prime}\subseteq c$ and
$1(c^{\prime},a_{i})+1(c^{\prime},a_{i^{\prime}})=2$, $\left \vert \overline
{S}_{i,i^{\prime}}\cap c^{\prime}\right \vert \geq2$.\hfill$\clubsuit$

(This is because of how we defined $\widetilde{j}$ and $\widetilde{c}$ above:
for $j\in a_{i^{\prime}}\cap c$, agent $\widetilde{j}$ is selected from within
the set $a_{i}\cap \widetilde{c}$ where $\widetilde{c}\subseteq c^{\prime}$;
thus, $\widetilde{j}\neq j$ and $j,\widetilde{j}\in \overline{S}_{i,i^{\prime}%
}\cap c^{\prime}$.)
\end{itemize}

\bigskip

We now apply Theorem \ref{thm:2color} by setting $N^{\prime
}=\overline{S}_{i,i^{\prime}}$: There exist $N_{B}^{\prime}\subseteq
\overline{S}_{i,i^{\prime}}$ and $N_{W}^{\prime}=\overline{S}_{i,i^{\prime}%
}\smallsetminus \overline{S}_{i,i^{\prime}}$ such that $\left \vert
N_{B}^{\prime}\right \vert =k$ and for each $\widetilde{c}\in C$, $\left \vert
\left \vert N_{B}^{\prime}\cap \widetilde{c}\right \vert -\left \vert
N_{W}^{\prime}\cap \widetilde{c}\right \vert \right \vert \leq1$. Let $a^{\prime
}$ be the allocation such that $a_{i}^{\prime}=\left(  a_{i}\smallsetminus
c\right)  \cup N_{W}^{\prime}$ and $a_{i^{\prime}}^{\prime}=\left(
a_{i}\smallsetminus c\right)  \cup N_{B}^{\prime}$. Clearly, at $a^{\prime}$
we have $\left \vert a_{i}^{\prime}\right \vert =q_{i}-1$ and $\left \vert
a_{i^{\prime}}^{\prime}\right \vert =q_{i^{\prime}}$. Hence, by Theorem \ref{thm:ExistenceNash},
$a^{\prime}$ is a Nash equilibrium. Note that $\gamma \left(  a^{\prime
},C\right)  \geq \gamma \left(  a,C\right)  $ because:

\begin{itemize}
\setlength\itemsep{1em}
\item[-] For each $c^{\prime}\in C$ such that $c^{\prime}\cap c=\emptyset$,%
\[
1\left(  c^{\prime},a_{i}\right)  +1\left(  c^{\prime},a_{i^{\prime}}\right)
=1\left(  c^{\prime},a_{i}^{\prime}\right)  +1\left(  c^{\prime},a_{i^{\prime
}}^{\prime}\right)  \text{.}%
\]

(Because agents in $N\smallsetminus c$ are allocated to resources in exactly
the same way at allocations $a$ and $a^{\prime}$.)

\item[-] For each $c^{\prime}\in C$ such that $c\subseteq c^{\prime}$,%
\[
1\left(  c^{\prime},a_{i}^{\prime}\right)  +1\left(  c^{\prime},a_{i^{\prime}%
}^{\prime}\right)  \geq1\left(  c^{\prime},a_{i}\right)  +1\left(  c^{\prime
},a_{i^{\prime}}\right)  \text{.}%
\]

(Because $\overline{S}_{i,i^{\prime}}\subseteq c$, $\left \vert \overline
{S}_{i,i^{\prime}}\right \vert \geq3$, and hence, by application of Theorem \ref{thm:2color} we obtain that $1\left(  c^{\prime}%
,a_{i}^{\prime}\right)  =1\left(  c^{\prime},a_{i^{\prime}}^{\prime}\right)
=1$.)

\item[-] For each $c^{\prime}\in C$ such that $c^{\prime}\subset c$ and
$\left \vert c^{\prime}\right \vert =1$,%
\[
1\left(  c^{\prime},a_{i}^{\prime}\right)  +1\left(  c^{\prime},a_{i^{\prime}%
}^{\prime}\right)  =1\left(  c^{\prime},a_{i}\right)  +1\left(  c^{\prime
},a_{i^{\prime}}\right)  =1\text{.}%
\]

(Because at any allocation a single agent is assigned to exactly one resource.)

\item[-] For each $c^{\prime}\in C$ such that $c^{\prime}\subset c$ and
$\left \vert c^{\prime}\right \vert \geq2$,%
\[
1\left(  c^{\prime},a_{i}^{\prime}\right)  +1\left(  c^{\prime},a_{i^{\prime}%
}^{\prime}\right)  \geq1\left(  c^{\prime},a_{i}\right)  +1\left(  c^{\prime
},a_{i^{\prime}}\right)  \text{.}%
\]

(Because: If $c^{\prime}\subseteq \left(  a_{i}\cap c\right)  $ or $c^{\prime
}\subseteq \left(  a_{i^{\prime}}\cap c\right)  $, we get $1\left(  c^{\prime
},a_{i}\right)  +1\left(  c^{\prime},a_{i^{\prime}}\right)  =1$ and the
desired result follows. If $1(c^{\prime},a_{i})+1(c^{\prime},a_{i^{\prime}%
})=2$, then $\left \vert \overline{S}_{i,i^{\prime}}\cap c^{\prime}\right \vert
\geq2$. (See the bullet argument above indicated with $\clubsuit$.) Hence the
desired result follows by application of Theorem \ref{thm:2color}.
\end{itemize}

Note that $\beta(a^{\prime})=\alpha +\beta_{i}$, $\beta(a)= \alpha +\beta_{i^{\prime}}$,
and since $\beta_{i}<\beta_{i^{\prime}}$, we get $\beta(a^{\prime})<\beta(a)$.
Since $\gamma \left(  a^{\prime},C\right)  \geq \gamma \left(  a,C\right)  $ and
$\beta(a^{\prime})<\beta(a)$, we obtain that $a^{\prime} $ $\gamma \beta
$-dominates $a$ w.r.t. $C$. Thus, as required, allocation $a^{\prime}$ is a
Nash equilibrium and $a^{\prime}$ $\gamma \beta$-dominates $a$ w.r.t. $C$.

\bigskip

This concludes our proof.
\end{proof}

\section*{Appendix C}

\begin{proof}[Proof of Theorem \ref{teo_two_noCOE}]
Consider an RSG where:

\begin{itemize}
\item[-] $N=\{1,\ldots,6\}$ and $M=\{1,2\}$.

\item[-] $f_{1}$ is such that $f_{1}(1)=1,f_{1}(2)=2,f_{1}(3)=4$.

\item[-] $f_{2}$ is such that $f_{2}(1)=1,f_{2}(2)=2,f_{2}(3)=3,f_{2}(4)=4$.
\end{itemize}

\medskip

Let $C=\mathcal{P}_{=1}(N)\cup\left\{  c_{1},c_{2},c_{3},c_{4},c_{5}\right\}  $ be such that the coalitions are as illustrated below.

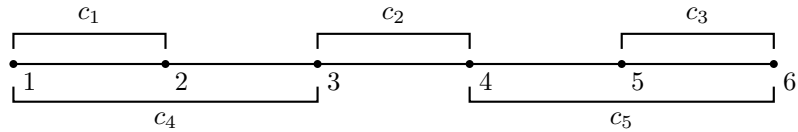
\begin{center}
\begin{tikzpicture}
\draw[thick] (0, 0) -- (10, 0);
\fill[black] (0, 0) circle (1.5pt) node[anchor = north west]{1};
\fill[black] (2, 0) circle (1.5pt) node[anchor = north west]{2};
\fill[black] (4, 0) circle (1.5pt) node[anchor = north west]{3};
\fill[black] (6, 0) circle (1.5pt) node[anchor = north west]{4};
\fill[black] (8, 0) circle (1.5pt) node[anchor = north west]{5};
\fill[black] (10, 0) circle (1.5pt)node[anchor = north west]{6};
\draw (1, 0.65) node{$c_1$};
\draw[thick] (0, 0.20) -- (0, 0.40) -- (2, 0.40) -- (2, 0.20);
\draw (5, 0.65) node{$c_2$};
\draw[thick] (4, 0.20) -- (4, 0.40) -- (6, 0.40) -- (6, 0.20);
\draw (9, 0.65) node{$c_3$};
\draw[thick] (8, 0.20) -- (8, 0.40) -- (10, 0.40) -- (10, 0.20);
\draw (2, -0.75) node{$c_4$};
\draw[thick] (0, -0.30) -- (0, -0.50) -- (4, -0.50) -- (4, -0.30);
\draw (8, -0.75) node{$c_5$};
\draw[thick] (6, -0.30) -- (6, -0.50) -- (10, -0.50) -- (10, -0.30);
\end{tikzpicture}
\end{center}

It is clear from the figure that $C\in\mathcal{C}^{coe}$. \ In this game we
will show that no allocation is $C$-stable. By way of contradiction, suppose
that there exists a $C$-stable allocation $a$.

\medskip

Since $\mathcal{P}_{=1}(N)\subset C$, $a$ is a Nash equilibrium. Using Theorem \ref{thm:ExistenceNash}, we obtain that there
are two possibilities: $\left\vert a_{1}\right\vert =3$ and $\left\vert
a_{2}\right\vert =3$ or $\left\vert a_{1}\right\vert =2$ and $\left\vert
a_{2}\right\vert =4$. Note that at $a$, it must be that at most one agent in
$c_{1}$ is assigned to the high resource. Otherwise, at $a$ coalition $c_{1}$
has a profitable deviation: If one agent in the coalition deviates to the
other resource, the other resource now becomes high. Thus, the well-being of
the agent that deviates remains the same (it is still assigned to a high
resource) while the other agent (now assigned to a low resource) becomes
better off. The same argument applies for coalitions $c_{2}$ and $c_{3}$.

\medskip

Note that if resource $2$ is high at $a$ (i.e., $\left\vert a_{2}\right\vert
=4$), it must be that $c_{1}\subset a_{2}$ or $c_{2}\subset a_{2}$ or
$c_{3}\subset a_{2}$. We showed that this cannot be true. Therefore, at $a$
resource $1$ is high and resource $2$ is low (i.e., $\left\vert a_{1}%
\right\vert =3$ and $\left\vert a_{2}\right\vert =3$).

\medskip

Since resource $1$ is high, we cannot have $c_{1}\subset a_{1}$ or
$c_{2}\subset a_{1}$ or $c_{3}\subset a_{1}$. But then, since $\left\vert
a_{1}\right\vert =3$, it must be that $\left\vert a_{1}\cap c_{1}\right\vert
=\left\vert a_{1}\cap c_{2}\right\vert =\left\vert a_{1}\cap c_{3}\right\vert
=1$. Consider coalition $c_{2}$. If $3\in a_{1}$, we obtain that $c_{4}$ is
such that $\left\vert a_{1}\cap c_{4}\right\vert =2$ and $\left\vert a_{2}\cap
c_{4}\right\vert =1$. If $4\in a_{1}$, we obtain that $c_{5}$ is such that
$\left\vert a_{1}\cap c_{5}\right\vert =2$ and $\left\vert a_{2}\cap
c_{5}\right\vert =1$. Wlog., suppose that the former case is true. But then at
$a$ consider the following deviation for $c_{4}$: Each agent in $c_{4}$
deviates to the other resource. The deviation makes resource $1$ low and
resource $2$ high. Note that at the induced allocation the agents that
deviated to resource $2$ are equally well-off (they are still assigned to a
high resource) and the agent that deviated to resource $1$ is better off
(because now it assigned to a low resource for which the beta value is
smaller). But then this is a profitable deviation, a contradiction.\hfill
\end{proof}

\begin{proof}[Proof of Theorem \ref{teo_twoident_noCEE}]
Consider an RSG where $N=\{1,\ldots,5\}$, $M=\{1,2\}$, and the two resources
are identical. Let $C=\mathcal{P}_{=1}(N)\cup
\{\{1,2\},\{3,4\},\{5,4\},\{1,2,3,5\},\{5,2,3,4\}\}$. Consider the planar
representation illustrated below. In the figure, arrows indicate the circles:
An arrow's tail indicates the circle's center. And its length is the radius of
the circle. (To keep the figure simple, coalitions of size 1 are not
indicated.) This planar representation is in accordance with coalition
structure $C$, i.e., $C\in\mathcal{C}^{cee}$.

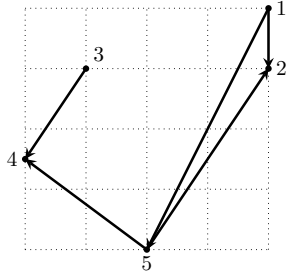
\begin{center}
\scalebox{0.8}{
\begin{tikzpicture}
\draw[dotted] (0, 0) -- (4, 0);
\draw[dotted] (0, 1) -- (4, 1);
\draw[dotted] (0, 2) -- (4, 2);
\draw[dotted] (0, 3) -- (4, 3);
\draw[dotted] (0, 4) -- (4, 4);
\draw[dotted] (0, 0) -- (0, 4);
\draw[dotted] (1, 0) -- (1, 4);
\draw[dotted] (2, 0) -- (2, 4);
\draw[dotted] (3, 0) -- (3, 4);
\draw[dotted] (4, 0) -- (4, 4);
\fill[black] (4, 4) circle (1.5pt) node[anchor = west]{1};
\fill[black] (4, 3) circle (1.5pt) node[anchor = west]{2};
\fill[black] (1, 3) circle (1.5pt) node[anchor = south west]{3};
\fill[black] (0, 1.5) circle (1.5pt) node[anchor = east]{4};
\fill[black] (2, 0) circle (1.5pt) node[anchor = north]{5};
\draw[very thick, ->, >=stealth] (4, 4) -- (4, 3);
%\draw (4, 4) circle (1);
\draw[very thick, ->, >=stealth] (4, 4) -- (2, 0);
%\draw (4, 4) circle (4.47);
\draw[very thick, ->, >=stealth] (1, 3) -- (0, 1.5);
%\draw (1, 3) circle (1.80);
\draw[very thick, ->, >=stealth] (2, 0) -- (4, 3);
%\draw (2, 0) circle (2.50);
\draw[very thick, ->, >=stealth] (2, 0) -- (0, 1.5);
%\draw (2, 0) circle (3.60);
\end{tikzpicture}}
\end{center}

In this game we will show that no allocation is $C$-stable. By way of
contradiction, suppose that there exists an allocation $a$ such that $a$ is
$C$-stable.

\medskip

Since $\mathcal{P}_{=1}(N)\subset C$, $a$ is a Nash equilibrium. Using Theorem \ref{thm:ExistenceNash} above, we obtain that at $a$
one resource is assigned two agents and the other one is assigned three
agents. Wlog., let $\left\vert a_{1}\right\vert =2$ and $\left\vert
a_{2}\right\vert =3$.

\medskip

Suppose that $1\in a_{1}$. Then, in $\{5,2,3,4\}$ there is one agent assigned
to resource $1$ and there are three agents assigned to resource $2$. But then
at $a$ the coalition $\{5,2,3,4\}$ has a profitable deviation: Wlog., let $a_{1}=\left\{  1,2\right\}  $ and $a_{2}=\{3,4,5\}$. At $a$, if agent $2$
deviates to resource $2$ and agents $3,4$ deviate to resource $1$, agent $5$
becomes better off and the well-beings of the remaining agents in coalition
$\{5,2,3,4\}$ do not change. This contradicts that $a$ is $C$-stable. Thus,
$1\in a_{2}$. Now suppose that $4\in a_{1}$. But then the preceding arguments
can be repeated for coalition $\{1,2,3,5\}$, leading to a contradiction. Thus,
$4\in a_{2}$. Therefore, $a_{2}$ is $\{1,2,4\}$ or $\{1,3,4\}$ or $\{1,4,5\}$.
But then we obtain that for coalitions $\left\{  1,2\right\}  $ or $\left\{
3,4\right\}  $ or $\left\{  5,4\right\}  $, at least one of them is a subset
of $a_{2}$. But then at $a$ this coalition has a profitable deviation: If one
agent in the coalition deviates to resource $1$, this agent's well-being
remains the same while the other agent in the coalition becomes better off.
This contradics that $a$ is $C$-stable. Therefore, $a$ is not $C$%
-stable.\hfill
\end{proof}

\end{document}